\documentclass[11pt,letterpaper]{article}
\usepackage[margin=1in]{geometry}
\usepackage{microtype}
\usepackage[utf8]{inputenc}
\usepackage{graphicx} 
\usepackage{tikz}
\usepackage{gensymb}  
\usepackage{enumitem}
\usepackage{amsmath}
\usepackage{amssymb}
\usepackage{amsthm}
\usepackage{thm-restate}
\usepackage{marvosym}  

\usepackage[colorlinks,allcolors=blue]{hyperref}
\usepackage{cleveref}
\newtheorem{theorem}{Theorem}
\newtheorem{fact}[theorem]{Fact}
\newtheorem{lemma}[theorem]{Lemma}
\newtheorem{corollary}[theorem]{Corollary}
\newtheorem{problem}[theorem]{Problem}

\begin{document}

\title{Near-Linear Time Maximum Clique in Circular Arc Graphs}

\author{Dario Moschetti \\ {\small University of Milan}
  \and Adam Polak \\ {\small Bocconi University}
  \and Jonas Schmidt \\ {\small Bocconi University}}

\date{}

\maketitle

\begin{abstract}
We present a near-linear time algorithm to find a maximum clique in circular arc graphs, the intersection graphs of arcs on the circumference of a circle, answering a 30-year-old open question posed by Bhattacharya and Kaller [J. Algorithms 1997] and Tiskin [CPM 2006]. Our algorithm runs in time $O(n\log^2 n)$, given $n$ circular arcs, where the graph is implicit. Previous approaches take at least quadratic time (unless the implicitly given graph is sparse). To overcome the quadratic barrier, we solve (structured and related) instances of bipartite matching by interpreting them as a single, offline 2D dominance problem, which in turn we model as a shortest paths problem in a planar grid graph. This allows for the final step, which builds on the theory of Monge matrices to obtain sparse representation and efficient distance computations in the grid graph.

Additionally, we provide an observation about the structure of circular arc cliques that leads to an algorithm for the weighted version of the problem with a new parameter, the minimum load $d$, which we define as the minimum, over all points on the circle, of the number of arcs containing that point. The algorithm runs in time $O(dn\log\log n + n\log n)$. This is always at least as good as the previous best maximum weight clique algorithm, which runs in time $O(m\log\log n + n\log n)$, where $m$ denotes the number of edges in the implicit intersection graph, because the minimum load satisfies $d \le 2m/n$.

The high level concepts behind our main algorithm for circular arc graphs also apply to a related graph class, interval filament graphs, which leads to an improved, randomized $O(n^\omega)$-time algorithm for maximum clique, where $\omega < 2.372$ is the matrix multiplication exponent. Last, we provide a quadratic fine-grained lower bound, conditional on the Orthogonal Vectors Hypothesis, which even applies to co-comparability graphs, a subclass of interval filament graphs.
\end{abstract}

\section{Introduction}\label{sec:intro}

Finding a maximum size clique in a graph is one of the most classical problems in theoretical computer science and belongs to Karp's list of 21 NP-complete problems~\cite{Karp72}.
That is why people study whether the maximum clique problem remains hard on restricted graph classes. One family of graph classes particularly studied in this context are \emph{geometric intersection graphs}, that is, graphs where every vertex corresponds to a region in a geometric space and two vertices are connected if and only if their corresponding regions intersect.
A lot of previous work has already identified geometric objects for which maximum clique remains NP-hard, such as straight line segments in the plane~\cite{cabello2013clique}, and many others where it becomes polynomial time solvable, such as axis-aligned rectangles in the plane~\cite{imai1983finding} or chords inside a circle~\cite{gavril1973algorithms,Tiskin06,Tiskin15}.
It is common to assume that the input is presented as $n$ geometric objects and the graph is implicit. 

One such famous graph class where the maximum clique problem is particularly easy to solve is the class of \emph{interval graphs}, where vertices correspond to intervals on the real number line.
Finding a maximum clique in interval graphs is even a common coding interview question: It can be solved in time $O(n \log n)$ by a simple sweep-line approach that maintains the set of active intervals. One of these sets must form a maximum clique.
A fundamental reason for why this simple algorithm works is that intervals satisfy the \emph{Helly property}.
In the language of geometric intersection graphs, the Helly property asserts that any set of vertices forming a clique must share a common point in the geometric model.
Intuitively, this restricts the maximum clique to a polynomial---or in case of interval graphs linear---number of candidates, which can be easily constructed and compared in size.
The same observation, with slightly more algorithmic effort, leads to a near-linear time algorithm for the intersection graph of rectangles in the plane~\cite{imai1983finding}.

What about geometric objects where the Helly property fails?
Arguably the simplest generalization of interval graphs that no longer satisfies the Helly property are \emph{circular arc graphs}, where vertices correspond to arcs on the circumference of a circle instead of intervals on the line.
These do not satisfy the Helly property because the arcs of a clique can wrap around the circle and meet on the other side instead of in a single common point. 
Circular arc graphs have been studied for a variety of algorithmic problems~\cite{garey1980complexity,hsu1985maximum,ApostolicoH87,ShihH89,hsu1991linear,shih1992n,BhattacharyaK97,mcconnell2003linear}, and they naturally model periodic scheduling tasks and other applications, see, e.g., \cite{golumbic2004algorithmic}.

After a line of improvements from the initial $O(nm)$ running time~\cite{hsu1985maximum,ApostolicoH87,ShihH89}, the best known algorithm for maximum clique in circular arc graphs is from the 1990s due to Bhattacharya and Kaller~\cite{BhattacharyaK97} and runs in time $O(m + n\log n)$, where $m$ is the number of edges in the implicitly given graph.
While a notable improvement, the latter running time still degrades to $n^2$ in dense graphs, i.e., when each circular arc overlaps on average with at least a constant fraction of other arcs.

This raises the natural question of whether the running time could be improved when measuring the time just in terms of $n$, the number of arcs.
Bhattacharya and Kaller~\cite{BhattacharyaK97} specifically ask if an improved algorithm could lead to a running time of $O(n \log n)$.
In the 30 years since their paper, no new algorithms for maximum clique in circular arc graphs have been developed.
The problem of improving the running time, however, was repeated in a 2006 paper by Tiskin~\cite{Tiskin06}, which studies cliques in intersection graphs of chords inside a circle.

\subsection{Our contribution}\label{sec:contribution}

We fully resolve the time complexity of maximum clique in circular arc graphs, up to logarithmic factors, by providing a near-linear time algorithm for the problem.

\begin{restatable}{theorem}{thmAlgCircularArcs}\label{thm:clique-circ-arc}
    Maximum clique in circular arc graphs can be solved in time $O(n\log^2n)$.
\end{restatable}

The techniques behind our algorithm that establishes \Cref{thm:clique-circ-arc} are noticeably different than previous algorithms for problems on circular arc graphs.
On a very high level, our algorithm solves (certainly structured and related) instances of bipartite matching by interpreting them as a 2D dominance problem, which we model in turn as a shortest paths problem in a planar grid graph.
Computing the paths in the last step crucially relies on \emph{Monge matrices}.
These are structured matrices that naturally encode distances in planar graphs~\cite{FakcharoenpholR06}, and are used to solve string problems~\cite{GaneshKLS22,GorbachevK25} and longest increasing subsequence~\cite{GawrychowskiGK25}.
Monge matrices have also been applied to the maximum clique problem in circle graphs via a reduction to a longest common subsequence query problem~\cite{Tiskin06}.
In our case, we have to work with Monge matrices directly to compute the paths of interest with a divide-and-conquer approach.
In Section~\ref{sec:tech-overview}, we give more details and intuition on how the algorithm behind \Cref{thm:clique-circ-arc} works.

\paragraph{Parameterization by minimum load.}
We provide another algorithm for maximum clique in circular arc graphs, parameterized by a new parameter, the \emph{minimum load}. 
For a point $\theta$ on the circle, define its load $d(\theta)$ to be the number of arcs passing through $\theta$.
The minimum load $d$ is the minimum of $d(\theta)$ over all points $\theta$ on the circle.
The algorithm is based on a structural property of circular arc graphs that we prove: In any given clique $C$, either (1) all arcs intersect in a single point, (i.e., the Helly property holds for this specific clique), or (2) the whole circumference is covered by the so-called \emph{base arcs} of $C$, i.e., arcs that do not contain other arcs of $C$. This algorithm also works for the weighted version of the problem, where each arc has a weight and the goal is to find a clique that maximizes the total weight.

\begin{restatable}{theorem}{thmParamCircularArcs}\label{thm:parameterized}
    The maximum weight clique problem in circular arc graphs can be solved in time $O(dn\log\log n + n\log n)$, where $d$ is the minimum load.
\end{restatable}

Note that the minimum load satisfies $d \le 2m/n$, where $m$ denotes the number of edges in the implicit intersection graph. Indeed, for any arc $a$, there must be at least $d$ other arcs at any point arbitrarily close to the head of $a$ but not contained in $a$, which implies that $a$ must have at least $d$ neighbors in the intersection graph. Summing over all arcs, there must be at least $dn/2$ edges.
This means that our new algorithm is always at least as good as the previous best maximum weight clique algorithm, which runs in time $O(m\log\log n + n\log n)$~\cite{ShihH89}.

\paragraph{Interval filament graphs.}
Conceptually, one way to view our algorithm for circular arc graphs from \Cref{thm:clique-circ-arc} is that it improves upon the following weakness of previous approaches:
After using the Helly property and a sweep-line approach for the ``interval part'' of the graph, the problem turns into $O(n)$ (structured) instances of bipartite matching, which are solved in sequence.
Our algorithm instead solves all subproblems together (using Monge matrices).
Part of our speedups can thus be explained by the increased flexibility of a fully \emph{offline} approach in contrast to an \emph{online/dynamic} approach.

We apply the same concept to another generalization of interval graphs that do not satisfy the Helly property, namely \emph{interval filament graphs}.
In these graphs, every vertex has an associated interval on the $x$-axis and a 2D curve living in the region above the interval, attached to both endpoints of the interval. Two vertices are adjacent if their curves intersect.
Note that in a clique the corresponding intervals must still share a common point on the $x$-axis; however,
this is not sufficient, the corresponding curves still have to intersect, and each pair can intersect at a different point, hence the Helly property fails.

Interval filament graphs were introduced by Gavril~\cite{Gavril00}, who gave a polynomial-time algorithm for maximum clique, running in time $O(n^4)$.\footnote{Since Gavril's algorithm uses a black-box bipartite matching algorithm as a subroutine, the running time immediately becomes $n^{3+o(1)}$ by plugging in the recent almost-linear time flow algorithm~\cite{ChenKLPGS25}; however, to improve the algorithm beyond cubic time, new ideas seem necessary.}
Other algorithmic problems on interval filament graphs have been studied~\cite{Pergel07,BestW21}, but the clique algorithm has not been improved.
Gavril's algorithm follows the same structure as discussed above:
Each interesting point on the $x$-axis yields one candidate for the maximum clique, which is found by solving a bipartite matching instance.
We show that an offline dynamic bipartite matching algorithm can be used to again exploit that we know all the subproblems upfront, leading to a significantly improved algorithm.
Here, $\omega < 2.372$ is the matrix multiplication exponent~\cite{AlmanDWXXZ25}.

\begin{restatable}{theorem}{thmAlgIntervalFilament}\label{thm:filament-clique}
    Maximum clique in interval filament graphs can be solved in time $O(n^\omega)$ with a randomized algorithm that is correct with high probability.
\end{restatable}

Even if $\omega = 2$, which is still an open possibility~\cite{AlmanW23}, \Cref{thm:filament-clique} would not get us past the quadratic-time barrier. Is there any hope of breaking that barrier? Our algorithm in \Cref{thm:filament-clique}, like Gavril's~\cite{Gavril00}, abstracts away from the specific geometry of curves and assumes that the input consists of (1) the explicit intersection graph of the curves, and (2) the interval corresponding to the curve of each vertex in the intersection graph.
As long as the intersections (i.e., graph edges) are provided in the input, both algorithms therefore work for arbitrarily complex curves.
In this input model, improving the running time beyond quadratic is unconditionally impossible, because the number of edges may be quadratic.
However, if the curves are simple---say polygonal lines with only a few segments---and given explicitly in the input, while the graph is implicit, one could hope to exploit their geometry to solve the problem faster. Note that such a setting is strictly simpler, since we can always compute the intersection graph ourselves with line segment intersections.
We show that this is also (conditionally) impossible, by proving a quadratic fine-grained lower bound.
Our lower bound is conditional on the Orthogonal Vectors Hypothesis (OVH), a standard hypothesis from fine-grained complexity, which is implied by the Strong Exponential Time Hypothesis (SETH)~\cite{Williams05}. The lower bound is a consequence of the following more general lower bound for co-comparability graphs, which form a subclass of interval filament graphs. Recall that the co-comparability graph of a partially ordered set (poset) $P=(X, \leq)$ is the undirected graph $G=(X,E)$ with $E=\{(u, v) \in X \times X : \neg (u \leq v \lor v \leq u\}$.

\begin{restatable}{theorem}{thmLBCocomp}\label{thm:poset-lower}
    Unless OVH fails, there is no $O(n^{2-\epsilon})$-time algorithm that, given a $k$-dimensional poset $P \subseteq \mathbb{R}^k$ of $n$ points, solves maximum clique in the co-comparability graph of $P$, for any constant $\epsilon > 0$ and dimension $k = c\log n$, for a constant $c$ that depends only on $\epsilon$.
\end{restatable}

Since every co-comparability graph of a $k$-dimensional poset can be modeled as an intersection graph of interval filaments where each filament is a polygonal curve with $O(k)$ segments, the lower bound follows.

\begin{restatable}{corollary}{corLBIntervalFilament}\label{cor:lb-interval-filament}
    Unless OVH fails, maximum clique in interval filament graphs cannot be solved in time $O(n^{2-\epsilon})$, for any constant $\epsilon > 0$, even when each filament is given explicitly as a polygonal line with $c\log n$ segments, for a constant $c$ that depends only on $\epsilon$.
\end{restatable}

The lower bound further extends to $k$-trapezoid graphs, another superclass of co-comparability graphs (of dimension $k$).
A \emph{$k$-trapezoid} picks $k$ intervals on $k$ parallel lines, and joins consecutive intervals into a polygon across the lines. The intersection graph of a family of $k$-trapezoids is called a \emph{$k$-trapezoid graph}~\cite{Flotow95,FelsnerMW97}.

\begin{restatable}{corollary}{corLBktrapezoid}\label{cor:lb-trapezoid}
    Unless OVH fails, maximum clique in $k$-trapezoid graphs cannot be solved in time $O(n^{2-\epsilon})$, for any constant $\epsilon > 0$ and $k = c\log n$, for a constant $c$ that depends only on $\epsilon$.
\end{restatable}

This result confirms the intuition of Felsner, M\"uller, and Wernisch that the case $k > 2$ is more difficult than $k=2$, for which they provide an $O(n \log n)$-time algorithm~\cite{FelsnerMW97}.

Let us also mention that in the case of both co-comparability graphs and $k$-trapezoid graphs, the maximum clique problem reduces, via a constructive proof of Dilworth's theorem (see, e.g., \cite{HarPeledJ20,FelsnerMW97}), to the bipartite maximum matching problem, which can be solved in time $n^{2+o(1)}$~\cite{ChenKLPGS25}. This means that \Cref{thm:poset-lower} and Corollary~\ref{cor:lb-trapezoid} (unlike Corollary~\ref{cor:lb-interval-filament}) provide lower bounds that are tight even if $\omega > 2$.

See Section \ref{sec:lower-bounds} for proofs of these lower bounds.

\subsection{Technical overview}\label{sec:tech-overview}

\paragraph{The radial sweep of Bhattacharya and Kaller.}

\begin{figure}[t]
    \centering
    \begin{tikzpicture}[
  dot/.style={circle,fill,inner sep=1.25pt},
  hit/.style={line width=1.75pt,line cap=round},
]

\newcommand{\circumference}{\draw[dashed] (0,0) circle (1);}
\newcommand{\arcA}[1]{\draw[#1] (150:1.15) arc[start angle=150,end angle=260,radius=1.15];}
\newcommand{\arcB}[1]{\draw[#1] (200:1.30) arc[start angle=200,end angle=330,radius=1.30];}
\newcommand{\arcC}[1]{\draw[#1] (320:1.45) arc[start angle=320,end angle=430,radius=1.45];}
\newcommand{\arcD}[1]{\draw[#1] (355:1.15) arc[start angle=355,end angle=470,radius=1.15];}
\newcommand{\arcE}[1]{\draw[#1] (40:1.30) arc[start angle=40,end angle=165,radius=1.30];}
\newcommand{\arcF}[1]{\draw[#1] (105:1.45) arc[start angle=105,end angle=235,radius=1.45];}
\newcommand{\allblack}{\arcA{}\arcB{}\arcC{}\arcD{}\arcE{}\arcF{}}
\newcommand{\ptA}[2]{\node[dot] at (#1:1) {}; \node at ({#1+25}:.58) {#2};}
\newcommand{\ptB}[2]{\node[dot] at (#1:1) {}; \node at (#1:.58) {#2};}

\foreach \x/\a/\hits in {
  0/210/{\arcC{hit,blue}\arcD{hit,blue}\arcA{hit,red}\arcB{hit,red}\arcF{hit,red}},
  4/265/{\arcD{hit,blue}\arcE{hit,blue}\arcB{hit,red}},
  8/310/{\arcE{hit,blue}\arcF{hit,blue}\arcB{hit,red}}
}{
  \begin{scope}[shift={(\x,4)}]
    \circumference
    \allblack
    \draw[gray] (\a:1)--({\a+180}:1);
    \hits
    \ptA{\a}{$\theta$}
    \ptA{\a+180}{$\bar\theta$}
  \end{scope}
}

\foreach \x/\a/\hits in {
  0/240/{\arcA{hit,red}\arcB{hit,red}},
  4/325/{\arcB{hit,red}\arcC{hit,red}},
  8/0/{\arcC{hit,red}}
}{
  \begin{scope}[shift={(\x,0)}]
    \circumference
    \allblack
    \draw[gray] (90:1)--(\a:1);
    \arcD{hit,blue}\arcE{hit,blue}
    \arcA{red}\arcB{red}\arcC{red}\arcF{red}
    \hits
    \ptB{90}{$\phi$}
    \ptB{\a}{$\theta$}
  \end{scope}
}
\end{tikzpicture}
    \caption{Comparison of the radial sweep procedures. Top row: sweeping diameter, both points move~\cite{BhattacharyaK97}. Bottom row: sweeping chord, only one point moves (ours).}
    \label{fig:sweep}
\end{figure}
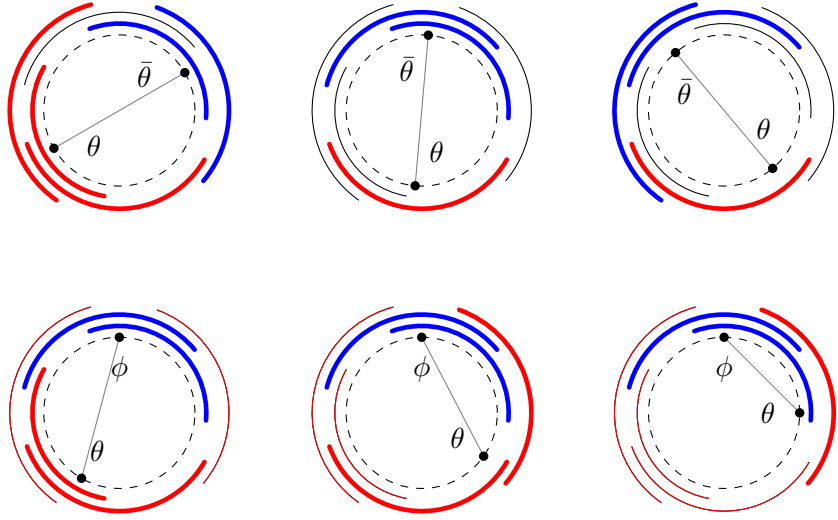

Our algorithm for maximum clique in circular arc graphs (\Cref{thm:clique-circ-arc}) uses a radial sweep method inspired by the work of Bhattacharya and Kaller~\cite{BhattacharyaK97} (see Figure~\ref{fig:sweep}, top row). Let us begin by briefly recalling their approach. They show that, for every clique, there exists a diameter of the circle such that each arc of the clique contains (at least) one of the endpoints of the diameter.\footnote{One can think of it as a relaxation of the Helly property: for every circular arc clique, there are \emph{two} points such that every clique arc contains one of them. Bhattacharya and Kaller show that it is enough to look at pairs of antipodal points. Later we prove that one point can be completely fixed.}
Therefore, a maximum clique can be found by sweeping a point $\theta$ around the circle, and, at each time step, looking for a maximum clique in the union of two sets of arcs: those containing $\theta$ and those containing its antipodal point $\bar\theta := \theta + 180\degree$. Each of these two sets is a clique itself, hence finding a maximum clique in their union is an instance of the maximum clique problem in \emph{co-bipartite} graphs---which is equivalent to the maximum independent set problem in bipartite graphs, which in turn can be solved using a bipartite matching algorithm thanks to K{\H o}nig's theorem. However, in order to get their claimed running time, Bhattacharya and Kaller cannot afford to solve the resulting $O(n)$ instances of bipartite matching independently using an off-the-shelf algorithm. Instead they devise a dynamic algorithm exploiting additional geometric structure. Here we take a different approach.

As we explain further in this section, we model the aforementioned co-bipartite clique problem as a path-finding problem involving points in the 2D plane. At each time step of the sweep, we need to solve a subproblem restricted to a rectangular region of the plane that correctly captures the arcs currently intersecting the sweep points. Here we face an issue: in the original sweep process of Bhattacharya and Kaller, both points $\theta$ and $\bar\theta$ move. This corresponds to all four sides of the rectangle moving. Using our techniques, combined with a baby-step giant-step approach, we could get an algorithm with the total running time $O(n^{1.5} \log^2 n)$. In order to get our claimed near-linear running time, we modify the sweep so that only one point moves (i.e., we sweep a chord instead of a diameter), and hence only two sides (one horizontal and one vertical) of the rectangular region move.

\paragraph{Our radial sweep.}

On the circle, we fix an arbitrary\footnote{All our statements indeed hold for an arbitrary $\phi$, but it is perhaps easier to think of $\phi$ that is not an endpoint of any arc, and since $\phi$ can be chosen arbitrarily, such a choice can be made without loss of generality.} point $\phi$. We say that an arc is \emph{blue} if it contains $\phi$, and that an arc is \emph{red} if it does not contain $\phi$ (see Figure~\ref{fig:sweep}, bottom row). Notice that blue arcs form a clique. Moreover, red arcs form an interval graph---their intersections can be modeled as intersections of real-line intervals---so they have the Helly property, and all the maximal cliques among them can be enumerated with a simple line-sweep. This is what our algorithm does.\footnote{As a side note for the curious and imaginative reader, our radial sweep can be simulated with the Bhattacharya--Kaller sweep with a small preprocessing that consists of stretching the circle at point $\phi$ so as to triple its circumference.} We sweep a point $\theta$ around the circle, from $\phi$ to $\phi$, and, at each time step of the sweep, we look for a maximum clique among all the red arcs containing $\theta$ (which form a clique) and all the blue arcs (which also form a clique). This is still a co-bipartite clique problem, but, as we explain next, we model it geometrically in the 2D plane.

\paragraph{Compatible bichromatic sets in the plane.}

In order to model interactions between red and blue arcs, we map each arc to a point in the 2D plane (Figure~\ref{fig:overview-transformations}). For red arcs, we associate heads with the $x$-axis and tails with the $y$-axis, and for the blue arcs we do the opposite. The $x$ coordinates grow in the clockwise direction, the $y$ coordinates in the anti-clockwise. We color all points according to the colors of their arcs. We say that a point $(x,y)$ \emph{dominates} a point $(x',y')$ iff $x > x'$ and $y > y'$. By definition, a set of arcs that form a clique corresponds to a set of points in which no red point dominates any blue point (see Figures~\ref{fig:arcs-to-points} and~\ref{fig:dominated}). We call such a set of points \emph{compatible}.

\begin{figure}
    \centering
    \includegraphics[width=.95\linewidth]{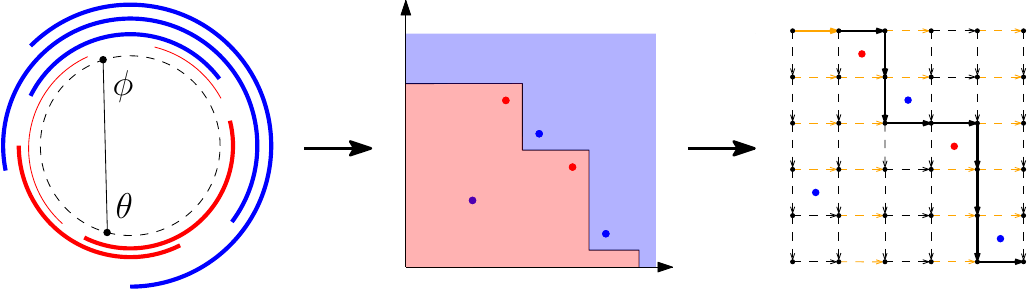}
    \caption{An overview of our approach. After the radial sweep (left), we interpret the problem as a 2D dominance problem, which is equivalent to finding a staircase with a maximum number of red points below and blue points above (middle). We then model the staircase problem as finding a shortest path in a planar grid graph, where edges hold positive weight if traversing the edge means missing out on a point, depicted in orange (right).}
    \label{fig:overview-transformations}
\end{figure}

A compatible set $C$ can also be represented as an axis-aligned, right- and downward monotone path through the plane just above every red point in $C$.
Since no blue point in $C$ is dominated by a red point in $C$, all blue points in $C$ must reside in the region above the path.
We call such a path a \emph{staircase}.
Moreover, every staircase also induces a compatible set by picking all red points below and blue points above it.
Hence, we can equivalently look for a staircase with maximum number of red points below plus blue points above it (see \Cref{fig:overview-transformations}, middle).

An important part of the above reduction to finding maximum compatible sets is that every arc is assigned a point independent of the position of $\theta$ in the radial sweep.
Moreover, to each sweep position $\theta$, we can assign a point $(x_\theta,y_\theta)$ in our coordinate system such that the red arcs containing $\theta$ correspond exactly to the red points with $x$-coordinate less than $x_\theta$ and $y$-coordinate less than $y_\theta$.
The subproblems hence ask about the maximum compatible set (or maximum staircase) in the quadrant $(-\infty, x_\theta] \times (-\infty, y_\theta]$ of the plane.
Summarizing the discussion so far, what remains is to solve the following problem.

\begin{restatable}[Maximum compatible bichromatic set quadrant queries]{problem}{probQuadrant}\label{prob:quadrant}
We are given two sets of \emph{red} and \emph{blue} points in the plane $R, B \subseteq \mathbb{R}^2$, with $n$ points in total, and $n$ query points $(X_1, Y_1), \ldots, (X_n, Y_n) \in \mathbb{R}^2$. For each query $i \in [n]$, output the size of a maximum compatible subset of the points lying in the quadrant $(-\infty, X_i] \times (-\infty, Y_i]$.
\end{restatable}

In order to solve \Cref{prob:quadrant} in near-linear time, we preprocess the set of all points offline to obtain a data structure that allows us to answer the quadrant queries in just logarithmic time.
To this end, we perform yet another interesting transformation: from the view of compatible sets or staircases in quadrants of the plane to shortest paths in a grid graph.

\paragraph{The grid graph.}

One way to go about finding a staircase that induces a maximum compatible bichromatic set of points is to model the problem as a graph problem:
Vertices correspond to relevant cells in the 2D plane, and edges point from a cell to the neighboring cells to the right and below.
A staircase in a lower-left quadrant of the plane then corresponds to a path from the top-left cell to the bottom-right cell in the region.
By setting positive edge weights for horizontal edges with red points above or blue points below, we make sure that the weight of a path corresponds to the number of points excluded from a solution.
This means that a maximum compatible set corresponds to a shortest path in the \emph{grid graph} just described (see \Cref{fig:overview-transformations}, right).

At first glance, this construction has an obvious downside:
The grid graph has $\Theta(n^2)$ vertices and edges, let alone that we have to find shortest paths between $n$ distinct pairs of vertices to answer all queries. How should this help to achieve a near-linear time algorithm?
This is where \emph{Monge matrices} enter the stage.
These are structured matrices that, among other things, capture the structure of shortest paths between vertices on the outer face of a planar graph~\cite{FakcharoenpholR06}. Our grid graph is clearly planar and our desired shortest path queries have their endpoints on its outer face.
Furthermore, they allow for a sparse representation of the distance matrix, in our case only of size $O(n)$, the original number of points in the plane we considered.
Related graph constructions for string alignment problems are also known as \emph{alignment graphs}~\cite{benson1995space,schmidt1998all,landau2001common,Tiskin06,russo2012monge,GawrychowskiGK25,GorbachevK25,boneh2025bounded}.
All that remains is to efficiently compute the representation of the shortest paths we are interested in. We achieve this by a divide-and-conquer approach along the $x$-axis as well using known fast min-plus product algorithms for Monge matrices with sparse representations~\cite{GawrychowskiGK25}.

\subsection{Open problems}

\paragraph{Maximum weight clique in circular arc graphs.}
Apostolico and Hambrusch~\cite{ApostolicoH87}
posed as an open problem the weighted version of the maximum clique in circular arc graphs, where each arc carries a weight and we seek a clique of maximum total weight. Shih and Hsu~\cite{ShihH89} settled it with an $O(n \log n + m \log\log n)$-time algorithm, which was later matched by Bhattacharya and Kaller~\cite{BhattacharyaK97} with a different technique, and still remains the fastest known.
In principle, our main algorithm extends to the weighted version of the problem, but the efficiency guarantees do not. We discuss in Section \ref{sec:weighted} that the distance matrices remain Monge, but their sparse representations are no longer guaranteed to be small. Proving any subquadratic bound on their size, and hence on the time needed to build them, turns out to be nontrivial: We formally argue in \Cref{thm:lis-barrier} that it would already give a faster data structure for the weighted range query version of the longest increasing subsequence (LIS) problem, solving a notorious open question, which remains open because of similar sparsity-related obstacles~\cite{GawrychowskiGK25}.
A first question is to understand whether our algorithm actually degrades to quadratic running time when handling weighted instances, or whether the inefficiency is only in the analysis. In any case, it would be interesting to understand whether maximum weight clique in circular arc graphs can be solved in subquadratic time.

\paragraph{Maximum clique in interval filament graphs.}
For maximum clique in interval filament graphs, our upper bound (\Cref{thm:filament-clique}) and lower bound (Corollary~\ref{cor:lb-interval-filament}) do not match, unless $\omega = 2$. It is interesting whether a quadratic-time algorithm exists even if $\omega > 2$. A related (and perhaps easier) problem is to get any truly subcubic running time without resorting to algebraic techniques for fast matrix multiplication.

\paragraph{Shaving the log.}
Our algorithm for circular arc graphs works in time $O(n \log^2 n)$. Could it be improved down to $O(n \log n)$? We do not see any inherent hardness of the problem preventing a better algorithm. In particular, we would not be very surprised if some more clever divide-and-conquer approach could be used to compute the same Monge matrix that our algorithm computes but in time $O(n \log n)$, which would be sufficient for the entire algorithm to run in time $O(n \log n)$. Having said that, we do not have any immediate idea how to achieve such a speedup. On the other hand, $\Omega(n \log n)$ seems to be a lower bound, at least for real-valued inputs in the comparison-addition model, where this time is already required to decide if $n$ given numbers are all distinct, which easily reduces to finding a clique of size $2$ in a circular arc graph.

\subsection{Organization of the paper}
In Section~\ref{sec:prelim} we provide basic definitions and notational conventions.
In Sections \ref{sec:structure}--\ref{sec:building} we present subsequent building blocks of our main algorithm for maximum clique in circular arc graphs. We put them together in Section~\ref{sec:wrap-up} and prove \Cref{thm:clique-circ-arc}.
In Section~\ref{sec:weighted} we discuss obstacles in generalizing our approach to the weighted version of the problem.
In Section~\ref{sec:parameterized} we give our second algorithm, parameterized by the minimum load, proving \Cref{thm:parameterized}. Then, we move on to interval filament graphs: in Section~\ref{sec:filaments} we present our algorithm (\Cref{thm:filament-clique}), and in Section~\ref{sec:lower-bounds} we prove the lower bounds (Theorem~\ref{thm:poset-lower} and Corollaries~\ref{cor:lb-interval-filament} and~\ref{cor:lb-trapezoid}).

\section{Preliminaries}
\label{sec:prelim}

Let $[n] = \{1, \dots, n\}$. 
We say that a randomized algorithm with input size $n$ is correct \emph{with high probability} if it returns a correct answer with probability at least $1 - 1/n^c$ for an arbitrarily large prespecified constant $c$.

Since our algorithms can always output the arcs/filaments in a maximum clique, we define the problems in their reporting version.

\paragraph{Circular arc graphs.}
A \emph{circular arc model} is a set $\mathcal{S}$ of $n$ arcs on a circle. Given an arc $a$, we define the \emph{head} of $a$ as the endpoint where the arc begins when traversing it clockwise, and the \emph{tail} of $a$ as the endpoint where it ends. They are contained in the arc.
As is common in the literature on geometric intersection graphs, we assume without loss of generality that no two arcs share a common endpoint.
Two arcs intersect if they contain a common point.
A \emph{clique} in a circular arc model $\mathcal{S}$ is a subset $\mathcal{C} \subseteq \mathcal{S}$ such that every pair of arcs in $\mathcal{C}$ intersect.

\begin{problem}[Maximum clique in circular arc graphs]
    Given a circular arc model $\mathcal{S}$, output a maximum size clique in $\mathcal{S}$.
\end{problem}

\paragraph{Interval filament graphs.}
An \emph{interval filament} $f$ is a continuous curve with endpoints $(\ell_f, 0)$ and $(r_f, 0)$ contained entirely inside the region $[\ell_f, r_f] \times [0, \infty) \subseteq \mathbb{R}^2$ in the plane.
We say that $[\ell_f, r_f]$ is the \emph{spanning interval} for the filament $f$.
We can again assume without loss of generality that no two interval filaments share a common endpoint.

To allow for arbitrary curves, algorithms usually do not receive the set of $n$ curves $\mathcal{F} = \{f_1, \dots, f_n\}$ in the input. 
Instead, we assume the \emph{intersection graph} of $\mathcal{F}$ is given to us directly, namely a graph $G$ on $n$ vertices $[n]$, where $i$ and $j$ are adjacent if the curves $f_i$ and $f_j$ intersect.
To allow for the exploitation of Helly-like properties, we also assume that the spanning interval $[\ell_{f_i}, r_{f_i}]$ is known for each vertex $i \in [n]$.
This input model is standard for interval filament graphs~\cite{Gavril00,BestW21}.

\begin{problem}[Maximum clique in interval filament graphs]\label{prob:interval-filament}
    Given an intersection graph $G$ of a set of interval filaments and the spanning interval of the filament for each vertex, output a maximum size clique $C$ in $G$.
\end{problem}

Note that our lower bound in Corollary~\ref{cor:lb-interval-filament} applies to an easier variant of \Cref{prob:interval-filament} where the curves are restricted to polygonal lines and explicitly provided in the input. This would also allow the algorithm to compute the intersection graph itself, in time depending on the number of line segments.

\section{Structure of cliques in circular arc graphs}
\label{sec:structure}

We begin the description of our algorithm for maximum clique in circular arc graphs (\Cref{thm:clique-circ-arc}) by proving a structural result that guides our radial sweep approach.

\begin{lemma}
Let $\mathcal{S}$ be a set of circular arcs. Fix an arbitrary point $\phi$ on the circle. Partition the arcs into \emph{blue} arcs that contain $\phi$ and \emph{red} arcs that do not contain $\phi$, $\mathcal{S} = \mathcal{B} \cup \mathcal{R}$. For any red arc $r \in \mathcal{R}$, let $\mathcal{R}_r \subseteq \mathcal{R}$ denote the set of all red arcs that contain the head of $r$. For every clique $\mathcal{C} \subseteq \mathcal{S}$, either $\mathcal{C} \subseteq \mathcal{B}$, or there exists a red arc $r \in \mathcal{R}$ such that $\mathcal{C} \subseteq \mathcal{B} \cup \mathcal{R}_r$.
\end{lemma}

\begin{proof}
Let $\mathcal{C} \subseteq \mathcal{S}$ be a set of arcs that form a clique. Suppose $\mathcal{C} \not\subset \mathcal{B}$, i.e., the clique contains at least one red arc. Consider heads of the red arcs of the clique $\mathcal{R} \cap \mathcal{C}$ sorted in the clockwise order starting from $\phi$, and let $r$ denote the arc whose head is the last one in this order. We argue that $\mathcal{R} \cap \mathcal{C} \subseteq \mathcal{R}_r$, which implies $\mathcal{C} \subseteq \mathcal{B} \cup \mathcal{R}_r$, as wanted. Indeed, each arc $r' \in \mathcal{R} \cap \mathcal{C}$ intersect $r$, and, since the head of $r'$ lies before the head of $r$, it must be that the tail of $r'$ lies after the head of $r$ (with the last statement being true because neither $r$ nor $r'$ contains $\phi$). Hence, $r'$ contains the head of $r$, and therefore $r' \in \mathcal{R}_r$, as wanted.
\end{proof}

This structural result allows us to reduce finding a maximum clique in the whole input arc set into $|\mathcal{R}| \leq n$ subproblems of finding a maximum clique in the union of two cliques $\mathcal{B} \cup \mathcal{R}_r$.

\begin{corollary}
\[\max \{ |\mathcal{C}| : \mathcal{C} \subseteq \mathcal{S},\ \mathcal{C}\text{ is a clique} \} = \max\bigl\{|\mathcal{B}|, \max_{r \in \mathcal{R}} \max \{ |\mathcal{C}| : \mathcal{C} \subseteq \mathcal{B} \cup \mathcal{R}_r,\ \mathcal{C}\text{ is a clique} \}\bigr\}.\]
\end{corollary}

Later it becomes useful that (1) the first clique $\mathcal{B}$ is the same in all subproblems, and (2) the second clique $\mathcal{R}_r$ is a subset of $\mathcal{R}$ defined by a rectangular region in the 2D space where each arc is represented by a point whose coordinates correspond to that arc's head and tail.

\section{Reduction to quadrant queries}

In this section we reduce finding a maximum clique to $n$ queries to a single set of points in the plane. By the previous section, it suffices to compute, for the head $\theta$ of every red arc, a maximum clique among the blue arcs and the red arcs containing $\theta$, and to compare the largest of these cliques with the clique of all blue arcs. We turn each of these subproblems into one query about a dominance problem on points in the plane, which we now define.

\paragraph{The maximum compatible bichromatic set problem.}
\label{sec:mcbs}

Consider two sets of \emph{red} and \emph{blue} points in the plane $R, B \subseteq \mathbb{R}^2$. We say that a red point $r = (x_r,y_r) \in R$ \emph{dominates} a blue point $b = (x_b,y_b) \in B$ if $r$ lies strictly above and to the right of $b$, that is, $x_r > x_b$ and $y_r > y_b$. A subset $C \subseteq R \cup B$ is \emph{compatible} if no red point in $C$ dominates any blue point in $C$. The \emph{maximum compatible bichromatic set} problem asks, given the sets $R$ and $B$, for a compatible subset of $R \cup B$ of maximum size.

\begin{figure}[t!]
    \centering
    \begin{minipage}[c]{0.3\textwidth}
        \centering
        \includegraphics[page=1, scale=0.65]{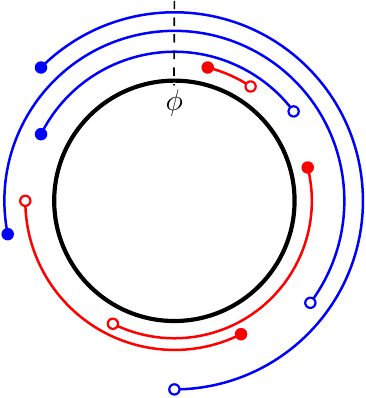}
    \end{minipage}\hfill
    \begin{minipage}[c]{0.345\textwidth}
        \centering
\begin{tikzpicture}[line width=1pt, scale=1.25, black!65]
 \node at (-.025,.2) {$\rotatebox{90}{\Leftscissors}$};
  \draw ([shift={(0,-0.28)}]108:0.28) arc[start angle=108, end angle=432, radius=0.28];
  \draw (0.5864,0) arc[start angle=90, end angle=-150, radius=0.28];
  \draw (1.1729,0) arc[start angle=90, end angle=-30, radius=0.28];
  \draw (0,0) -- (1.7593,0);
  \draw[->, line width=0.5pt, dash pattern=on 1pt off 1pt]
    plot[smooth] coordinates {(-0.086,-0.014) (0.086,-0.194) (0.344,-0.420)
      (0.880,-0.530) (1.415,-0.420) (1.600,-0.280) (1.73, -.1)};
\end{tikzpicture}
        \vskip .8cm
        \includegraphics[page=2, scale=0.65]{circ-arc.pdf}
    \end{minipage}\hfill
    \begin{minipage}[c]{0.345\textwidth}
        \centering
        \includegraphics[page=3, scale=0.65]{circ-arc.pdf}
    \end{minipage}
    \caption{Our mapping from arcs to points in the plane. Left: a circular arc model, with the dashed line marking the cut point $\phi$. Center: the same model cut open at $\phi$ and normalized. Right: the point set $S$, with one point per arc.}
    \label{fig:arcs-to-points}
\end{figure}

\paragraph{Cutting the circle open.}
\label{sec:normalized}

We fix a point $\phi$ of the circle; we call $\phi$ the \emph{cut point}. This partitions the set of arcs into the \emph{blue arcs} $\mathcal{B}$ containing $\phi$, and the red arcs $\mathcal{R}$ not containing $\phi$. It is useful to imagine ``cutting'' the circle at $\phi$ and ``unrolling'' it into a line. A red arc does not contain $\phi$, so it becomes a segment of the line, with its head to the left of its tail. A blue arc contains $\phi$, so it becomes two pieces, one at each end of the line, with its tail to the left of its head.

We further \emph{normalize} the model, reassigning the $2n$ arc endpoints to the positions $0, 1, \ldots, 2n - 1$ of the line, in their original order. This takes $O(n \log n)$ time, the time to sort the endpoints, and the intersection graph does not change, as it depends only on the order of the endpoints. We identify each arc with the pair $(h, t)$ of its head and tail positions, so that $h < t$ for a red arc and $t < h$ for a blue one.

\paragraph{From arcs to points.}
\label{sec:global}
To every arc we assign one point of the same color:
\[
    (h, t) \in \mathcal{R} \;\longmapsto\; \text{red point at } (h,\, 2n - t),
    \qquad
    (h, t) \in \mathcal{B} \;\longmapsto\; \text{blue point at } (t,\, 2n - h).
\]
In both cases the $x$-coordinate of the point is the left endpoint of the arc, and the $y$-coordinate is its right endpoint, measured from the right end of the line. The red points form the set $R$, the blue points the set $B$, and the point set $S = R \cup B$ has exactly one point per arc and is built in $O(n)$ time. The whole construction is illustrated in Figure~\ref{fig:arcs-to-points}.

\begin{fact}
\label{fact:per-theta}
A red arc and a blue arc intersect if and only if the point of the red arc does not dominate the point of the blue arc (see Figure~\ref{fig:dominated}).
\end{fact}
\begin{proof}
On the line, the red arc is the segment $[h_r, t_r]$ and the blue arc is the union of $[0, t_b]$ and $[h_b, 2n - 1]$. The two arcs are disjoint exactly when the segment falls strictly between the two pieces, that is, when $t_b < h_r$ and $t_r < h_b$. The first inequality states that the red point lies strictly to the right of the blue one, and the second, read as $2n - t_r > 2n - h_b$, that it lies strictly above. Disjointness is therefore exactly the dominance of the red point over the blue one.
\end{proof}

\begin{figure}[t!]
    \centering
\begin{minipage}[c]{0.5\textwidth}
    \centering
    \begin{tikzpicture}[
      scale=0.4,
      rh/.style={circle, fill=red, inner sep=1.5pt},
      rt/.style={circle, draw=red, fill=white, inner sep=1.1pt, line width=0.7pt},
      bh/.style={circle, fill=blue, inner sep=1.5pt},
      bt/.style={circle, draw=blue, fill=white, inner sep=1.1pt, line width=0.7pt},
      arc/.style={line width=0.9pt},
      bold/.style={line width=1.5pt}
    ]
      \draw (0,0) -- (13,0);
      \draw[arc, red] (0,1) -- (1,1);
      \draw[arc, red] (3,1) -- (7,1);
      \draw[bold, red] (5,2) -- (10,2);
      \draw[bold, blue] (0,2) -- (2,2);   \draw[bold,blue] (11,2) -- (13,2);
      \draw[arc, blue] (0,3) -- (4,3);   \draw[arc, blue] (9,3) -- (13,3);
      \draw[arc, blue] (0,4) -- (6,4);   \draw[arc, blue] (12,4) -- (13,4);
      \node[rh] at (0,1) {};  \node[rt] at (1,1) {};
      \node[rh] at (3,1) {};  \node[rt] at (7,1) {};
      \node[rh] at (5,2) {};  \node[rt] at (10,2) {};
      \node[bt] at (2,2) {};  \node[bh] at (11,2) {};
      \node[bt] at (4,3) {};  \node[bh] at (9,3) {};
      \node[bt] at (6,4) {};  \node[bh] at (12,4) {};
      \node[left] at (-0.2,4) {\scriptsize $b_1$};
      \node[left] at (-0.2,3) {\scriptsize $b_2$};
      \node[left] at (-0.2,2) {\scriptsize $\boldsymbol{b_3}$};
      \node[left] at (-0.2,1) {\scriptsize $r_1$};
      \node[right] at (7.2,1) {\scriptsize $r_2$};
      \node[left] at (4.8,2) {\scriptsize $\boldsymbol{r_3}$};
    \end{tikzpicture}
\end{minipage}\hfill
\begin{minipage}[c]{0.5\textwidth}
    \centering
    \begin{tikzpicture}[
      scale=0.35,
      rp/.style={circle, fill=red, inner sep=1.4pt},
      bp/.style={circle, fill=blue, inner sep=1.4pt},
    ]
      \fill[red!10] (-0.7,-0.7) rectangle (5,3);
      \draw[dashed, line width=0.5pt] (5,-0.7) -- (5,3) -- (-0.7,3);
      \draw[->] (-0.7,-0.7) -- (8.2,-0.7);
      \draw[->] (-0.7,-0.7) -- (-0.7,12.2);
      \node[rp] at (0,11) {}; \node[right] at (0.2,11) {\scriptsize $r_1$};
      \node[rp] at (3,5) {};  \node[right] at (3.2,5.3) {\scriptsize $r_2$};
      \node[rp,inner sep=2pt] at (5,3) {};  \node[right] at (5.2,3.4) {\scriptsize $\boldsymbol{r_3}$};
      \node[bp] at (4,4) {};  \node[right] at (4.2,4.4) {\scriptsize $b_2$};
      \node[bp,inner sep=2pt] at (2,2) {};  \node[right] at (2.2,2.3) {\scriptsize $\boldsymbol{b_3}$};
      \node[bp] at (6,1) {};  \node[right] at (6.2,1.3) {\scriptsize $b_1$};
    \end{tikzpicture}
\end{minipage}
    \caption{A red point dominates a blue point iff their corresponding arcs do not intersect. In the left, the red arc $r_3$ intersects every blue arc except $b_3$; in the right, the region dominated by the point of $r_3$ contains only the point of $b_3$.}
    \label{fig:dominated}
\end{figure}
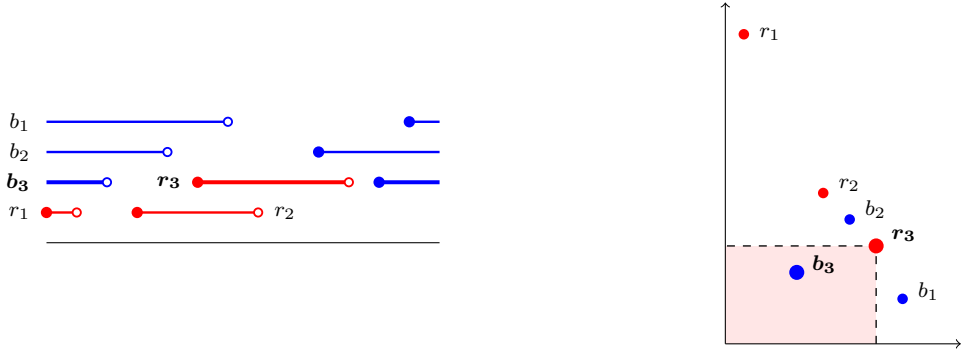

Fix a candidate $\theta$, the head of a red arc. The subproblem at $\theta$ asks for a maximum clique among the red arcs containing $\theta$ and all the blue arcs. The red arcs containing $\theta$ pairwise intersect at $\theta$, and the blue arcs pairwise intersect at the cut point, so it is only the intersections between red and blue arcs that determine whether a given subset of arcs forms a clique. Moreover, if a blue arc contains $\theta$, it intersects every arc under consideration, so it belongs to every maximum clique of the subproblem. Let $f_\theta$ be the number of these blue arcs. The rest of such a clique is a maximum clique among the red arcs containing $\theta$ and the blue arcs not containing $\theta$. The next lemma shows that these remaining arcs are exactly the arcs whose associated points lie in a quadrant of the plane.

\begin{lemma}
\label{lem:visibility}
Let $\theta$ be the head of a red arc. The points of $S$ in the quadrant $(-\infty, \theta] \times (-\infty, 2n - \theta]$ are exactly the points of the red arcs containing $\theta$ and of the blue arcs not containing $\theta$, and a set of these arcs is a clique if and only if their points form a compatible set.
\end{lemma}
\begin{proof}
A red arc contains $\theta$ exactly when $h \le \theta \le t$, and its point lies in the quadrant exactly when $h \le \theta$ and $2n - t \le 2n - \theta$, the same condition. A blue arc does not contain $\theta$ exactly when $t < \theta < h$, and its point lies in the quadrant exactly when $t \le \theta$ and $h \ge \theta$, again the same condition, since $\theta$ is an endpoint of no blue arc and the inequalities are strict. This proves the first part.

For the second part, the red arcs of the quadrant pairwise intersect at $\theta$, the blue arcs pairwise intersect at the cut point, and by Fact~\ref{fact:per-theta} a red and a blue arc intersect exactly when the point of the red arc does not dominate the point of the blue one. A set of these arcs is therefore a clique exactly when its points form a compatible set.
\end{proof}

A maximum clique of the subproblem at $\theta$ has therefore size $\mu_\theta + f_\theta$, where $\mu_\theta$ is the size of a maximum compatible subset of the points in the quadrant.

Note also that the points of $S$ have pairwise distinct $x$- and $y$-coordinates. The $x$-coordinates are the heads of the red arcs and the tails of the blue arcs, the $y$-coordinates are the complements to $2n$ of the tails of the red arcs and of the heads of the blue arcs, and the endpoint positions are pairwise distinct.

\paragraph{Quadrant queries.}
\label{sec:wmcbs}

Every candidate thus asks for a maximum compatible subset of a quadrant of the same point set, and the whole reduction rests on the following problem.

\probQuadrant*

An efficient solution gives an efficient algorithm for the maximum clique. We build $S$ once and, for the head $\theta$ of every red arc, we query the quadrant with $X = \theta$ and $Y = 2n - \theta$ and add the count $f_\theta$. The largest of these sums $\mu_\theta + f_\theta$, compared with the number of blue arcs, which form a clique themselves, is the size of a maximum clique. The next section solves the problem with $O(n \log^2 n)$ preprocessing time and $O(\log n)$ query time. The complete algorithm, including the computation of the counts $f_\theta$ and the recovery of an actual clique, is made precise in Section~\ref{sec:wrap-up}.

\section{Answering maximum compatible bichromatic set queries}
\label{sec:solving}

In this section we solve Problem~\ref{prob:quadrant}. Given the sets $R$ and $B$, we preprocess them so that, for any two reals $X$ and $Y$, the size of a maximum compatible subset of the points lying in the quadrant $(-\infty, X] \times (-\infty, Y]$ can be returned quickly. We write $S = R \cup B$ and $n = |S|$. Throughout, we assume that the points of $S$ have pairwise distinct $x$- and $y$-coordinates. The instance produced in Section~\ref{sec:global} satisfies the assumption by construction, and the next lemma shows that any instance can be brought to this form.

\begin{lemma}
\label{lem:distinct}
In $O(n \log n)$ time, the coordinates of the points of an instance can be reassigned so that they become pairwise distinct on each axis and the compatible sets are unchanged.
\end{lemma}
\begin{proof}
A red point dominates a blue point when both of its coordinates are strictly larger, so which sets are compatible depends only on how each red and blue pair compares on the two axes, and not on the coordinate values themselves. Sort the points by $x$-coordinate, placing each red point before every blue point with the same $x$-coordinate and ordering points of the same color arbitrarily, and assign to the $k$-th point of this order the new $x$-coordinate $k$. Reassign the $y$-coordinates in the same way. The new coordinates are pairwise distinct on each axis, and every strict comparison between old coordinates holds between the new ones as well. Consider a red and a blue point tied on an axis. The pair was free of dominance, and after the reassignment its red point has the smaller coordinate on that axis, so the pair stays free of dominance. Every red and blue pair therefore dominates after the reassignment exactly when it dominated before, and the compatible sets are unchanged. The cost is that of the two sorts.
\end{proof}

The \emph{shadow} of a point $q = (x_q, y_q)$ is the open lower-left quadrant $\{(x, y) : x < x_q,\, y < y_q\}$. The shadow of a set is the union of the shadows of its members. A \emph{downward staircase} is a non-increasing step function $\Sigma : \mathbb{R} \to \mathbb{R}$, where $\Sigma(x)$ denotes its $y$-value at $x$-coordinate $x$. The \emph{score} of a staircase $\Sigma$ is the number of blue points on or above $\Sigma$ plus the number of red points on or below $\Sigma$.

\begin{lemma}
\label{lem:staircase}
The size of a maximum compatible bichromatic set equals the maximum score over all downward staircases.
\end{lemma}
\begin{proof}
Let $C$ be a compatible set with red part $R_C$. Since no red point in $C$ dominates any blue one in $C$, no blue point of $C$ lies in the shadow of $R_C$. The upper-right boundary of this shadow is a downward staircase. Every red point of $C$ lies on or below it and every blue point of $C$ lies on or above it, so its score is at least $|C|$.

Conversely, for any staircase $\Sigma$, no red point on or below $\Sigma$ can dominate any blue point on or above $\Sigma$. Dominance requires both $x_r > x_b$ and $y_r > y_b$. If $x_r > x_b$, then by monotonicity $y_r \le \Sigma(x_r) \le \Sigma(x_b) \le y_b$, so $y_r \le y_b$ and dominance fails. Hence the set of points scored by $\Sigma$ is compatible.
\end{proof}

\subsection{The grid graph}
\label{subsec:graph}

We model the computation of a staircase of maximum score as a shortest path problem in a directed acyclic graph. A path in the graph describes the movements of a staircase, right and down, and its length grows by one every time the staircase fails to score a point. A shortest path then fails as rarely as possible, and by Lemma~\ref{lem:staircase} its length reveals the size of a maximum compatible set. We now introduce the graph.

The construction can be carried out on top of any sets of red and blue points with distinct coordinates, and Section~\ref{sec:building} will build it on parts of the input.

\paragraph{Grid rows and grid columns.}
The points of $S$ define $n$ vertical and $n$ horizontal lines, each passing through exactly one point of $S$, since no two points share a coordinate. The vertical lines cut the plane into $n+1$ open vertical strips, the \emph{grid columns}, and the horizontal lines cut it into $n+1$ open horizontal strips, the \emph{grid rows}. A \emph{cell} is the intersection of a grid column and a grid row, so there are $(n+1) \times (n+1)$ of them. We number the grid columns $0, 1, \ldots, n$ from left to right and the grid rows $0, 1, \ldots, n$ from bottom to top, and the matrices we define are indexed by these numbers.

\paragraph{The grid graph.}
The \emph{grid graph} $G$ has a vertex for every cell, a \emph{downward} edge of weight $0$ from every cell to the cell below it, and a \emph{rightward} edge from every cell to the cell to its right. A rightward edge in grid row $i$ crosses the vertical line of a single point $q$, and its weight is $0$ if $q$ is red and lies below grid row $i$ or blue and lies above grid row $i$, and $1$ otherwise. In other words, crossing the line of $q$ costs nothing exactly when a staircase crossing it inside grid row $i$ scores $q$. Figure~\ref{fig:grid-query} (left) illustrates the construction.

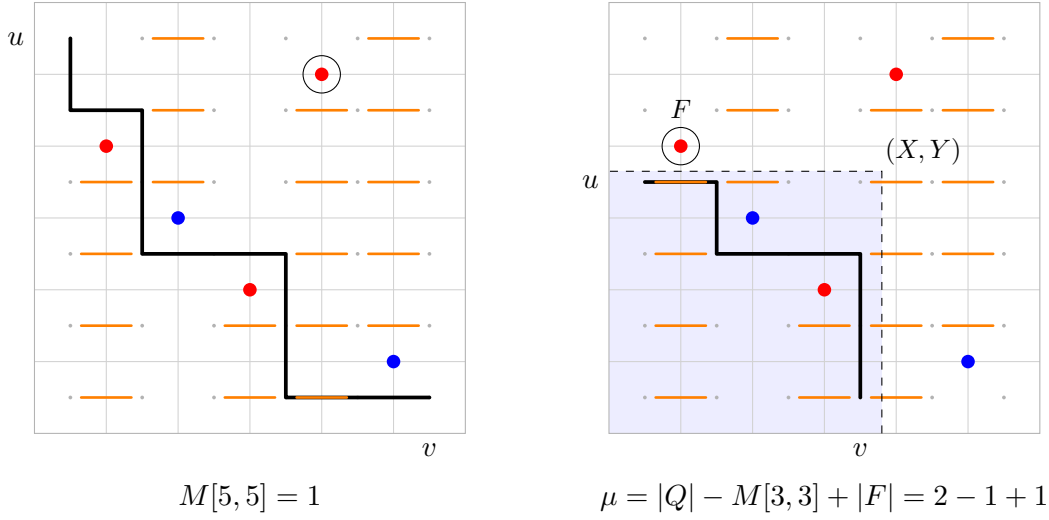
\begin{figure}
    \centering
\begin{tikzpicture}[
  scale=0.95,
  rpt/.style={circle,fill=red,inner sep=1.8pt},
  bpt/.style={circle,fill=blue,inner sep=1.8pt},
  cell/.style={circle,fill=gray!60,inner sep=0.5pt},
  wone/.style={orange,line width=1pt,line cap=round},
  spath/.style={line width=1.4pt,line cap=round,line join=round},
]

\newcommand{\gridpanel}{
  \foreach \i in {1,...,5} {
    \draw[gray!35] (\i,0) -- (\i,6);
    \draw[gray!35] (0,\i) -- (6,\i);
  }
  \draw[gray!60] (0,0) rectangle (6,6);
  \foreach \x/\y in {1/4, 3/2, 4/5}
    \foreach \r in {1,...,\y} \draw[wone] (\x-0.35,\r-0.5) -- (\x+0.35,\r-0.5);
  \foreach \x/\y in {2/3, 5/1}
    \foreach \r in {\y,...,5} \draw[wone] (\x-0.35,\r+0.5) -- (\x+0.35,\r+0.5);
  \foreach \c in {0,...,5} \foreach \r in {0,...,5}
    \node[cell] at (\c+0.5,\r+0.5) {};
  \node[rpt] at (1,4) {}; \node[rpt] at (3,2) {}; \node[rpt] at (4,5) {};
  \node[bpt] at (2,3) {}; \node[bpt] at (5,1) {};
}

\begin{scope}
  \gridpanel
  \draw[spath] (0.5,5.5) -- (0.5,4.5) -- (1.5,4.5) -- (1.5,2.5)
    -- (3.5,2.5) -- (3.5,0.5) -- (5.5,0.5);
  \draw[wone] (3.65,0.5) -- (4.35,0.5); 
  \node[left] at (0,5.5) {$u$};
  \node[below] at (5.5,0) {$v$};
  \draw (4,5) circle (0.26);
  \node at (3,-0.9) {$M[5,5] = 1$};
\end{scope}

\begin{scope}[shift={(8,0)}]
  \fill[blue!7] (0,0) rectangle (3.8,3.65);
  \gridpanel
  \draw[dashed] (3.8,0) -- (3.8,3.65) -- (0,3.65);
  \node[above right, inner sep=1pt] at (3.8,3.7) {\small $(X,Y)$};
  \draw[spath] (0.5,3.5) -- (1.5,3.5) -- (1.5,2.5) -- (3.5,2.5) -- (3.5,0.5);
  \draw[wone] (0.65,3.5) -- (1.35,3.5); 
  \node[left] at (0,3.5) {$u$};
  \node[below] at (3.5,0) {$v$};
  \draw (1,4) circle (0.26);
  \node[above] at (1,4.25) {\small $F$};
  \node at (3,-0.9) {$\mu = |Q| - M[3,3] + |F| = 2 - 1 + 1$};
\end{scope}

\end{tikzpicture}
    \caption{The grid graph of a small point set. Left: the weight-$1$ edges, in orange, lie below each red point and above each blue point. The bold path realizes a staircase of maximum score, and its weight $M[5,5] = 1$ pays for the circled point it misses. Right: the query at $(X, Y)$, with $u$ in the grid row of $Y$ and $v$ in the grid column of $X$. The circled point is a false penalty: it lies outside the quadrant, yet every path from $u$ to $v$ crosses its line.}
    \label{fig:grid-query}
\end{figure}

The distances we need are collected in the \emph{prefix matrix} $M$, the $(n+1) \times (n+1)$ matrix whose entry $M[i, j]$ is the distance in $G$ from the cell in grid column $0$ and grid row $i$ to the cell in grid column $j$ and grid row $0$, that is, from the left side of the grid to its bottom side. Every entry is finite, since moving right along grid row $i$ and then descending reaches the target.

The following lemma, proved in Section~\ref{sec:building}, gives compact access to the prefix matrix, with the mapping of the query coordinates to the grid handled inside the structure. To this end we define the grid row and the grid column of a real. When no point of $S$ has $y$-coordinate $y$, the \emph{grid row of} $y$ is simply the grid row containing it. When instead a point has $y$-coordinate exactly $y$, the grid row of $y$ is the grid row just above that point. The \emph{grid column of} a real $x$ is defined in the same way: the grid column containing $x$, or the grid column just to the right of the point whose $x$-coordinate is $x$. The side is chosen with the queries in mind: a point lying exactly at a query value belongs to the closed quadrant, and it must remain possible to score it.

\begin{restatable}{lemma}{efficientshortestpaths}
\label{lem:efficient-shortest-paths}
Let $R$ and $B$ be sets of red and blue points with pairwise distinct $x$- and $y$-coordinates, $n$ points in total. There is a data structure, built in $O(n \log^2 n)$ time, that, given two reals $x$ and $y$, returns in $O(\log n)$ time the entry $M[i, j]$ of the prefix matrix, where $i$ is the grid row of $y$ and $j$ is the grid column of $x$.
\end{restatable}

\subsection{Answering the queries}
\label{sec:queries}

We can now state and prove the main result of the section.

\begin{lemma}
\label{lem:all-ell}
Let $R$ and $B$ be sets of red and blue points with pairwise distinct $x$- and $y$-coordinates, $n$ points in total. After a preprocessing step taking $O(n \log^2 n)$ time, given any reals $X$ and $Y$, the size of a maximum compatible subset of the points lying in $(-\infty, X] \times (-\infty, Y]$ can be returned in $O(\log n)$ time.
\end{lemma}

\begin{proof}
The preprocessing builds the data structure of Lemma~\ref{lem:efficient-shortest-paths} on the point set $R \cup B$ in $O(n \log^2 n)$ time, as well as two instances of a standard static two-dimensional orthogonal range counting data structure, one on $R$ and one on $B$. The latter data structure is built in $O(n \log n)$ time and, when given an axis-aligned rectangle as a query, takes $O(\log n)$ time to compute the number of points inside this rectangle~\cite{Willard85,Chazelle88}.

Fix a query, given by two reals $X$ and $Y$, let $i$ be the grid row of $Y$ and $j$ the grid column of $X$, let $Q$ be the set of points lying in the quadrant $(-\infty, X] \times (-\infty, Y]$, and let $\mu$ be the size of a maximum compatible subset of $Q$, the value to return. By Lemma~\ref{lem:staircase}, applied to the points of $Q$ alone, $\mu$ equals the largest score that a downward staircase attains on $Q$, that is $|Q|$ minus the smallest number of points of $Q$ that a staircase fails to score. We now express this last quantity through the prefix matrix, following the staircase along the grid.

Consider a path of $G$ from the cell $u$ in grid column $0$ and grid row $i$ to the cell $v$ in grid column $j$ and grid row $0$. Moving from grid column $0$ to grid column $j$, the path crosses exactly the vertical lines to the left of $v$, and by the definition of the grid column of $X$ these are the lines of the points with $x$-coordinate at most $X$. All the other points can be ignored. Consider now a point with $x$-coordinate at most $X$ but $y$-coordinate larger than $Y$. It lies outside $Q$, so it should not count, yet the path crosses its line. Such a point lies above grid row $i$, and the path, which starts in grid row $i$ and never rises, stays below it when crossing its line. When the point is blue this crossing has weight $0$ and is harmless. When it is red, the crossing has weight $1$, a \emph{false penalty} that every path from $u$ to $v$ incurs (Figure~\ref{fig:grid-query}, right). Collecting the false penalties in the set
\[
    F \;=\; \{r \in R : x_r \le X,\, y_r > Y\},
\]
the weight of every path from $u$ to $v$ is $|F|$ plus the number of points of $Q$ that it fails to score.

It remains to observe that the paths from $u$ to $v$ realize precisely the staircases of $Q$. In one direction, consider any downward staircase running through the open cells visited by the path. Whenever the path crosses the line of a red point above it, or the line of a blue point below it, the staircase scores that point. By Lemma~\ref{lem:staircase} applied to $Q$, the staircase scores at most $\mu$ points of $Q$, so the path fails on at least $|Q| - \mu$ of them. In the other direction, let $C$ be a maximum compatible subset of $Q$ and let $R_C$ be its red part. As in the proof of Lemma~\ref{lem:staircase}, the boundary of the shadow of $R_C$ separates $C$, and we build a path from $u$ to $v$ that follows it. Consider one of the crossed vertical lines, and let $r$ be the highest point of $R_C$ lying on that line or to its right. We cross the line in the grid row just above $r$, and in grid row $0$ when there is no such point. In this way the path stays above every point of $R_C$ whose line is still to be crossed, and drops as low as possible otherwise. The grid rows of the crossings are non-increasing, and they start at most at grid row $i$ because the points of $R_C$ have $y$-coordinate at most $Y$, so they define a path from $u$ to $v$. The path scores every red point of $C$, whose line it crosses just above it or higher. It also scores every blue point $b$ of $C$, since by compatibility every point of $R_C$ to the right of $b$ lies strictly below it, so the grid row of the crossing at $b$ lies below $b$. The path therefore fails on at most $|Q| - |C| = |Q| - \mu$ points of $Q$.

The smallest weight of a path from $u$ to $v$ is thus exactly $|F| + |Q| - \mu$, and this weight is the entry $M[i, j]$ of the prefix matrix. Rearranging,
\[
    \mu \;=\; |Q| \;-\; M[i, j] \;+\; |F|.
\]
The entry is returned by the data structure of Lemma~\ref{lem:efficient-shortest-paths}, and $|Q|$ and $|F|$ are counts of points in axis-aligned ranges, reported by the range counting structure, all in $O(\log n)$ time.
\end{proof}

In our algorithm the structure is built once on the point set $S$ of Section~\ref{sec:global}, and one query is asked for each of the $O(n)$ queried heads $\theta$, so the preprocessing and all the queries together take $O(n \log^2 n)$ time.

\section{Building the matrices}
\label{sec:building}

In this section we prove Lemma~\ref{lem:efficient-shortest-paths}, which we first restate.

\efficientshortestpaths*

We write $S = R \cup B$. Recall that the entry $M[i, j]$ of the prefix matrix is the distance in the grid graph from the cell in grid column $0$ and grid row $i$ to the cell in grid column $j$ and grid row $0$. We build the matrix by divide and conquer on the grid columns. The point set is split by a vertical line into a left and a right part, the matrices of the two parts are computed recursively, and the matrix of the whole set is assembled from them.

The assembly rests on a second matrix of distances. The \emph{across matrix} $N$ is the $(n+1) \times (n+1)$ matrix whose entry $N[i, j]$ is the distance in $G$ from the cell in grid column $0$ and grid row $i$ to the cell in grid column $n$ and grid row $j$, from the left side of the grid to its right side. A path that ends in the right part of the split enters it through the dividing line, so its left portion is measured by an across matrix of the left part, and the recursion computes the two matrices together. Since every edge of $G$ keeps or lowers the grid row, $N[i, j] = +\infty$ exactly when $i < j$. This forces one adjustment. The tools we build on store and multiply finite Monge matrices, and the infinite entries admit no compact representation, so we augment the graph with upward edges, which makes every entry finite while leaving the distances we need unchanged, and we work with the augmented graph throughout the construction. The across matrix is also of independent interest: at the end of the section we expose it through a data structure analogous to the one of Lemma~\ref{lem:efficient-shortest-paths}, which Section~\ref{sec:weighted} uses when discussing the weighted version of the problem.

\subsection{Core-sparse Monge matrices}
\label{sec:core-sparse}

A matrix $A$ is \emph{Monge} if $A[i, j] + A[i+1, j+1] \le A[i, j+1] + A[i+1, j]$ for all valid $i$ and $j$, equivalently $A[a, c] + A[b, d] \le A[a, d] + A[b, c]$ for all $a < b$ and $c < d$. Our construction stores and multiplies Monge matrices through the core-sparse machinery of~\cite{GawrychowskiGK25}. We recall the notions it rests on. The \emph{density matrix} of a matrix $A$ has entries $A^\square[i, j] = A[i, j+1] + A[i+1, j] - A[i, j] - A[i+1, j+1]$, and $A$ is Monge precisely when $A^\square$ is everywhere non-negative. The \emph{core} of $A$ is the set of positions where $A^\square$ is non-zero, the \emph{core size} $\delta(A)$ is their number, and the \emph{core sum} $\delta^\Sigma(A)$ is the sum of the corresponding density values. The \emph{condensed representation} of $A$ stores, for every position of the core, its row index, its column index, and its density value, together with the first row and the first column of $A$. These data determine every entry of $A$: by~\cite[Observation~3.1]{GawrychowskiGK25},
\[
    A[i, j] \;=\; A[i, 0] + A[0, j] - A[0, 0] - \delta^\Sigma(A'),
\]
where $A'$ is the submatrix of $A$ on the rows up to $i$ and the columns up to $j$, whose core sum collects the densities at the positions $(i', j')$ with $i' < i$ and $j' < j$.

The \emph{min-plus product} of a $p \times q$ matrix $A$ and a $q \times r$ matrix $B$ is the $p \times r$ matrix $A \oplus B$ with entries
\[
    (A \oplus B)[i, j] \;=\; \min_{\ell}\, A[i, \ell] + B[\ell, j].
\]

\begin{theorem}[\protect{\cite[Theorem~1.2]{GawrychowskiGK25}}]
\label{thm:ggk-product}
Let $A$ and $B$ be Monge matrices of sizes $p \times q$ and $q \times r$, given by their condensed representations. The min-plus product $A \oplus B$ is Monge, and its condensed representation can be computed in $O\bigl(p + q + r + (\delta(A) + \delta(B)) \log(1 + \delta(A) + \delta(B))\bigr)$ time.
\end{theorem}

\begin{lemma}[\protect{\cite[Lemma~4.4]{GorbachevK25}}]
\label{lem:gk-oracle}
Let $A$ be a matrix of size $p \times q$, given by its condensed representation. In $O\bigl(p + q + \delta(A) \log(1 + \delta(A))\bigr)$ time, the representation can be turned into an oracle that returns any entry of $A$ in $O\bigl(\log(2 + \delta(A))\bigr)$ time.
\end{lemma}

We bound the cores of our matrices through their entries. 

\begin{fact}
\label{fact:core-entries}
A Monge matrix $A$ with non-negative integer entries at most $U$ has a core of size at most $\delta(A) \le 2U$.
\end{fact}
\begin{proof}
Let $A$ have $p$ rows and $q$ columns, indexed from $0$. The core sum equals the alternating sum of the four corner entries, $\delta^\Sigma(A) = A[0, q-1] + A[p-1, 0] - A[0, 0] - A[p-1, q-1]$~\cite[Observation~3.1]{GawrychowskiGK25}, which is at most $2U$ when the entries lie between $0$ and $U$. When the entries are integers, the core size is at most the core sum, $\delta(A) \le \delta^\Sigma(A)$~\cite[Observation~B.10]{GawrychowskiGK25}.
\end{proof}

\subsection{The augmented grid graph}

The \emph{augmented grid graph} $\hat{G}$ is obtained from $G$ by adding an \emph{upward} edge of weight $1$ from every cell to the cell above it. We write $\hat{N}$ for the across matrix with distances taken in $\hat{G}$ instead. Every entry of $\hat{N}$ is finite, since the upward edges let a path reach any grid row. The next lemma shows that the augmentation creates no shortcut between cells that the grid graph already connects.

\begin{lemma}
\label{lem:no-upward}
Let $u$ be a cell in grid column $c_u$ and grid row $r_u$, and $v$ a cell in grid column $c_v$ and grid row $r_v$, with $c_u \le c_v$. If $r_u \ge r_v$, then there is a shortest path in $\hat{G}$ from $u$ to $v$ that does not use upward edges, i.e., the distance from $u$ to $v$ is the same in $G$ and in $\hat{G}$.
If $r_u < r_v$, some shortest path from $u$ to $v$ in $\hat{G}$ uses no downward edges.
\end{lemma}
\begin{proof}
We only prove the first claim for $r_u \ge r_v$, the argument for $r_u < r_v$ is the same.
Adding edges cannot increase distances, so it suffices to turn a shortest path $\pi$ of $\hat{G}$ from $u$ to $v$ into a path of $G$ of at most the same weight. Let $q_1, \ldots, q_m$ be the points whose vertical lines are crossed by $\pi$, from left to right, and let $\rho_1, \ldots, \rho_m$ be the grid rows of these crossings. Set
\[
    \rho'_k \;=\; \max\bigl(r_v,\, \min(r_u, \rho_1, \ldots, \rho_k)\bigr)
    \qquad \text{for } k = 1, \ldots, m.
\]
The sequence $\rho'_1, \ldots, \rho'_m$ is non-increasing, starts at most at $r_u$ and ends at least at $r_v$, so crossing the line of $q_k$ in grid row $\rho'_k$ for every $k$, with downward edges in between, defines a path from $u$ to $v$ in $G$.

It remains to compare the weight of the new path with the weight of $\pi$. The two paths cross the same vertical lines, and downward edges have weight $0$, so the weight of the new path is the sum of the weights of its crossings, while the weight of $\pi$ is the sum of the weights of its crossings plus the number of its upward edges. It therefore suffices to charge every crossing that weighs more in the new path to an upward edge of $\pi$ that crosses the horizontal line of the same point: the points $q_k$ are pairwise distinct, so distinct crossings charge distinct lines, hence distinct upward edges.

Fix $k$ such that the crossing of $q_k$ weighs more in the new path. The weight of a crossing depends only on whether its grid row is above or below the crossed point, so $\rho'_k \ne \rho_k$ and the line of $q_k$ separates grid row $\rho'_k$ from grid row $\rho_k$. Suppose first $\rho'_k < \rho_k$. Then $\min(r_u, \rho_1, \ldots, \rho_k)$ is at most $\rho'_k$, hence below the line of $q_k$, so $\pi$ ran below that line before its $k$-th crossing, at its start or at an earlier crossing, and rose back above it by the $k$-th crossing, through an upward edge crossing the horizontal line of $q_k$. Suppose now $\rho'_k > \rho_k$. Since $\min(r_u, \rho_1, \ldots, \rho_k) \le \rho_k$, this forces $\rho'_k = r_v$, so the line of $q_k$ lies above grid row $\rho_k$ and not above grid row $r_v$. The path $\pi$ crosses the line of $q_k$ in grid row $\rho_k$, below it, and ends in grid row $r_v$, above it, so it later rises through the line of $q_k$, again through an upward edge crossing it.
\end{proof}

\begin{fact}
\label{fact:monge-core}
The matrices $\hat{N}$ and $M$ are Monge, and their cores have size $O(n)$.
\end{fact}
\begin{proof}
The graph $\hat{G}$ is planar, and both matrices collect distances between vertices of its outer face: traversing it clockwise, the sources appear along the left side from bottom to top, followed by the targets in reverse order, along the right side from top to bottom for $\hat{N}$ and along the bottom side from right to left for $M$. This cyclic order guarantees that both matrices are Monge, as observed in~\cite[Section~2.3]{FakcharoenpholR06} and conveniently formalized in~\cite[Fact~B.7]{GawrychowskiGK25}.

We bound the entries. By Lemma~\ref{lem:no-upward} an entry of $M$ is realized by a path without upward edges, which crosses at most $n$ vertical lines at weight at most $1$ each, so the entry is at most $n$. For an entry $\hat{N}[i, j]$, a path may cross every line in grid row $i$ and then move vertically along grid column $n$ to grid row $j$, with weight at most $n$ from the crossings and at most $n$ from the climb, so the entry is at most $2n$. Both matrices are Monge with non-negative integer entries at most $2n$, so their cores have size at most $4n$ by Fact~\ref{fact:core-entries}.
\end{proof}

\subsection{The divide and conquer}
\label{sec:dc}

We compute the condensed representations of $\hat{N}$ and $M$ by divide and conquer. If $n = 1$, both matrices are $2 \times 2$ and we write them down directly. Otherwise we split the points of $S$ by a vertical line that passes through no point, into a left part and a right part whose sizes differ by at most one, and recurse on the two parts. The plan is the following: since no edge of $\hat{G}$ points leftward, a path that crosses the dividing line splits into a portion among the grid columns of the left part followed by a portion among the grid columns of the right part, so it suffices to know the distances of $\hat{G}$ restricted to the cells of each part, and a min-plus product combines the two portions. We therefore want, for each part, its across and prefix matrices with the distances taken in these restrictions of $\hat{G}$.

The recursion returns almost that. Each recursive call runs the same construction on the points of its part alone, as if they were the whole input, so its grid is cut only by the vertical and horizontal lines of these points. On the grid columns nothing is lost: every vertical line of a part is a vertical line of $S$, the two parts together provide all of them, and the dividing line lies inside a single grid column of each of the three grids, so that one grid column of $S$ corresponds to both the last grid column of the left part and the first grid column of the right part. The grid rows are more delicate. The grid of a part lacks the horizontal lines of the other part, so every grid row of a part is the union of one or more consecutive grid rows of $S$ (Figure~\ref{fig:dc}, left), and the matrices $\hat{N}_L$, $M_L$, $\hat{N}_R$, $M_R$ returned by the two calls are indexed by coarser grid rows, inconsistent with each other and with the grid of $\hat{G}$.

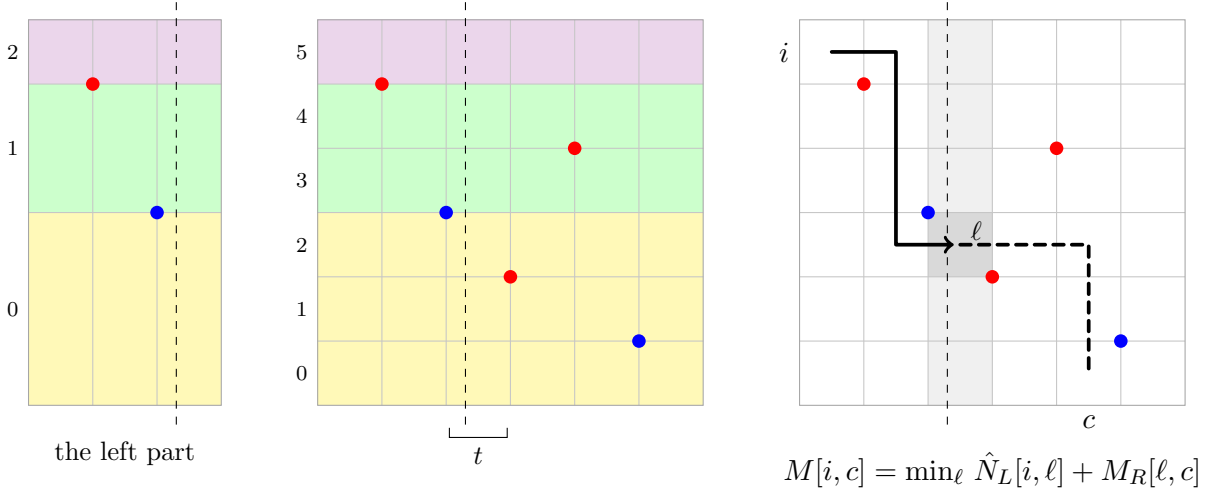
\begin{figure}
    \centering
\begin{tikzpicture}[
  scale=0.85,
  rpt/.style={circle,fill=red,inner sep=1.8pt},
  bpt/.style={circle,fill=blue,inner sep=1.8pt},
  lpath/.style={line width=1.4pt,line cap=round,line join=round},
]
\colorlet{rowa}{yellow!35}
\colorlet{rowb}{green!20}
\colorlet{rowc}{violet!16}

\begin{scope}[shift={(0,0)}]
  \fill[rowa] (0,0) rectangle (3,3);
  \fill[rowb] (0,3) rectangle (3,5);
  \fill[rowc] (0,5) rectangle (3,6);
  \foreach \x in {1,2} \draw[gray!45] (\x,0) -- (\x,6);
  \foreach \y in {3,5} \draw[gray!45] (0,\y) -- (3,\y);
  \draw[gray!70] (0,0) rectangle (3,6);
  \draw[dashed] (2.3,-0.3) -- (2.3,6.3);
  \node[rpt] at (1,5) {}; \node[bpt] at (2,3) {};
  \node[left] at (0,1.5) {\scriptsize $0$};
  \node[left] at (0,4.0) {\scriptsize $1$};
  \node[left] at (0,5.5) {\scriptsize $2$};
  \node[below] at (1.5,-0.4) {\small the left part};
\end{scope}

\begin{scope}[shift={(4.5,0)}]
  \fill[rowa] (0,0) rectangle (6,3);
  \fill[rowb] (0,3) rectangle (6,5);
  \fill[rowc] (0,5) rectangle (6,6);
  \foreach \i in {1,...,5} {
    \draw[gray!45] (\i,0) -- (\i,6);
    \draw[gray!45] (0,\i) -- (6,\i);
  }
  \draw[gray!70] (0,0) rectangle (6,6);
  \draw[dashed] (2.3,-0.3) -- (2.3,6.3);
  \node[rpt] at (1,5) {}; \node[rpt] at (3,2) {}; \node[rpt] at (4,4) {};
  \node[bpt] at (2,3) {}; \node[bpt] at (5,1) {};
  \foreach \r in {0,...,5} \node[left] at (0,\r+0.5) {\scriptsize $\r$};
  \draw (2.05,-0.35) -- (2.05,-0.5) -- (2.95,-0.5) -- (2.95,-0.35);
  \node[below] at (2.5,-0.5) {\small $t$};
\end{scope}

\begin{scope}[shift={(12,0)}]
  \fill[gray!12] (2,0) rectangle (3,6);
  \fill[gray!30] (2,2) rectangle (3,3);
  \foreach \i in {1,...,5} {
    \draw[gray!45] (\i,0) -- (\i,6);
    \draw[gray!45] (0,\i) -- (6,\i);
  }
  \draw[gray!70] (0,0) rectangle (6,6);
  \draw[dashed] (2.3,-0.3) -- (2.3,6.3);
  \node[rpt] at (1,5) {}; \node[rpt] at (3,2) {}; \node[rpt] at (4,4) {};
  \node[bpt] at (2,3) {}; \node[bpt] at (5,1) {};
  \draw[lpath,->] (0.5,5.5) -- (1.5,5.5) -- (1.5,2.5) -- (2.4,2.5);
  \draw[lpath,dash pattern=on 4pt off 3pt] (2.5,2.5) -- (4.5,2.5) -- (4.5,0.5);
  \node[left] at (0,5.5) {$i$};
  \node[below] at (4.5,0) {$c$};
  \node at (2.75,2.72) {\small $\ell$};
  \node at (3,-1) {$M[i,c] = \min_\ell\, \hat{N}_L[i,\ell] + M_R[\ell,c]$};
\end{scope}

\end{tikzpicture}
    \caption{Left: the dashed line splits the point set in grid column $t$, and each grid row of the left part, shown by the tints, is a union of consecutive grid rows of the whole grid. Right: a shortest path from the left side to the bottom side decomposes at the highlighted cell of grid column $t$ in grid row $\ell$ into an across path of the left part, solid, and a prefix path of the right part, dashed.}
    \label{fig:dc}
\end{figure}

We fix this by \emph{refining} the matrices, that is, re-indexing them by the grid rows of $S$ and adjusting their entries so that they hold the distances of $\hat{G}$ restricted to the grid columns of their part. Note that both indices of an across matrix are numbers of grid rows, at the left and at the right side of its grid, so refinement applies to both of them, while in a prefix matrix it applies to the first index only, as the second one ranges over grid columns. Restricted to the grid columns of a part, $\hat{G}$ is exactly the augmented grid graph of the part with every grid row subdivided by the horizontal lines of the other part, and subdividing changes the weight of no rightward edge, since the grid rows into which a grid row splits lie on the same side of every point of the part. For the prefix matrix the refinement is then plain repetition: by Lemma~\ref{lem:no-upward} its entries are realized by paths of $G$, which never climb, so the subdivision leaves them unchanged, and each entry is copied across the grid rows into which its grid row splits. For the across matrix the subdivision does change the entries, since its paths climb through upward edges, but the change is conveniently structured: Consider one subdividing line, splitting a grid row into the new grid rows $r-1$ and $r$. Merging the two grid rows back turns a path of the subdivided graph into a path of the coarser graph, saving one upward edge per rise through the line, and a path from a source below the line to a target above it rises through it at least once. Conversely, by Lemma~\ref{lem:no-upward}, a shortest path of the coarser graph can be chosen to cross the height of the line at most once, upward, and exactly when the source lies below the line and the target above it, so lifting it to the subdivided graph costs one extra upward edge in that case alone. The distances with source below and target above the line therefore grow by exactly one, while all the other distances are unchanged, and refining an across matrix amounts to repetition followed by adding $1$ to the quadrant $\{(i, j) : i < r \le j\}$ for each subdividing line.

Refining requires little work on the condensed representations. Consider first the repetition, which is applied to both matrices. The first row and the first column repeat their entries across the copies, and every core element stays in the core, with its density value, at a shifted position: between two consecutive copies of the same row the densities are zero, since the four terms of each density entry cancel in pairs, while the density entries between the last copy of a row and the first copy of the next keep their original values. One pass grows each stored index by the number of copies inserted at or before it.  Then we apply quadrant increments to the across matrix. Recall that an entry is determined by its first-column and first-row values and by the total density above and to its left, as mentioned in Section~\ref{sec:core-sparse}. Raising the first-row values at the columns $j \ge r$ by $1$ thus raises every entry with $j \ge r$ by $1$, and one new core element of density $1$ at the position $(r-1, r-1)$ lowers back by $1$ every entry with $i \ge r$ and $j \ge r$, so the two updates together add $1$ exactly on the quadrant $\{(i, j) : i < r \le j\}$. Each subdividing line thus contributes one interval update of the first row and one core element, merged with an existing element in case the positions coincide, and the whole refinement takes $O(n)$ time.

From now on we assume that the matrices of the two parts are refined to the grid rows of $S$, and the columns of their prefix matrices are numbered as the corresponding grid columns of $S$, so that all indices refer to the grid of $S$.

Let $t$ be the grid column of $S$ containing the dividing line. Because the graph has no leftward edge, a path from the left side of the grid to a cell on the right of grid column $t$ leaves grid column $t$ through a single rightward edge, at one grid row $\ell$, even though it may move vertically along grid column $t$ before doing so (Figure~\ref{fig:dc}, right). Cutting the path at that edge decomposes it into a path of the left part, from the left side to the cell in grid column $t$ and grid row $\ell$, and a path of the right part, from that cell to the right side, and the refined matrices measure exactly these two portions, so a shortest path from the left side in grid row $i$ to the right side in grid row $j$ weighs $\min_\ell \hat{N}_L[i, \ell] + \hat{N}_R[\ell, j]$. Entry by entry this reads
\[
    \hat{N} \;=\; \hat{N}_L \oplus \hat{N}_R,
\]
the min-plus product of the refined across matrices of the parts.

For the prefix matrix we work in $\hat{G}$ as well, which by Lemma~\ref{lem:no-upward} does not change its entries. A path from the left side to the bottom cell of a grid column $c \le t$ never leaves the grid columns of the left part, so $M[i, c] = M_L[i, c]$, and $M_L$ is the restriction of $M$ to the columns up to $t$. For $c > t$ the cut at grid column $t$ applies as before and gives
\[
    M[i, c] \;=\; \min_\ell\, \hat{N}_L[i, \ell] + M_R[\ell, c].
\]
The min-plus product $\hat{N}_L \oplus M_R$ therefore agrees with $M$ on the columns after $t$, and on column $t$ as well, so it is the restriction of $M$ to the columns from $t$ on.

The condensed representation of $M$ is now assembled from those of the two restrictions in $O(n)$ time. The density of $M$ at a position involves two consecutive columns, which lie together in one of the two restrictions, and $M$ agrees with that restriction on those columns, so the two densities coincide and the core of $M$ is the union of the two cores. Concretely, the stored representation of the product indexes its columns from $0$, so one pass over its core elements adds $t$ to their stored column indices, while the core of $M_L$ is kept as it is. The first row of $M$ is the first row of $M_L$ followed by the one of the product, the two agreeing at column $t$, and the first column of $M$ is the first column of $M_L$.

\paragraph{Running time.}
Consider any merge step of the recursion, and let $n'$ be the size of the merged set. The refinements take $O(n')$ time. The refined matrices have $O(n')$ rows and columns, and their cores have size $O(n')$ by Fact~\ref{fact:monge-core} applied to the parts. The two min-plus products then take $O(n' \log n')$ time by Theorem~\ref{thm:ggk-product}, and the assembly of $M$ takes $O(n')$ time. Since the split is balanced, the whole construction satisfies $T(n) = 2\,T(n/2) + O(n \log n)$, which solves to $T(n) = O(n \log^2 n)$.

\subsection{The data structure}

\begin{proof}[Proof of Lemma~\ref{lem:efficient-shortest-paths}]
The divide and conquer of Section~\ref{sec:dc} builds the condensed representation of $M$ in $O(n \log^2 n)$ time, together with the representation of $\hat{N}$ that it uses along the way. By Lemma~\ref{lem:gk-oracle}, in $O(n \log n)$ additional time we turn the representation of $M$ into an oracle that returns any entry in $O(\log n)$ time. We also store the $x$- and the $y$-coordinates of the points in sorted order, so that the grid row of a real $y$ and the grid column of a real $x$ are found by predecessor searches in $O(\log n)$ time. A query then makes two searches and one access to the oracle.
\end{proof}

From the same construction we also obtain a structure for the across matrix, which Section~\ref{sec:weighted} refers to when discussing the weighted version of the problem.

\begin{lemma}
\label{lem:across-queries}
Let $R$ and $B$ be sets of red and blue points with pairwise distinct $x$- and $y$-coordinates, $n$ points in total. There is a data structure, built in $O(n \log^2 n)$ time, that, given two reals $y$ and $y'$, returns in $O(\log n)$ time the entry $N[i, j]$ of the across matrix, equal to $+\infty$ when $i < j$, where $i$ is the grid row of $y$ and $j$ is the grid row of $y'$.
\end{lemma}
\begin{proof}
The divide and conquer of Section~\ref{sec:dc} builds the condensed representation of $\hat{N}$, and Lemma~\ref{lem:gk-oracle} turns it into an oracle answering in $O(\log n)$ time. A query finds the grid rows $i$ and $j$ by predecessor search, returns $+\infty$ if $i < j$, and reads $\hat{N}[i, j]$ otherwise, which equals $N[i, j]$ by Lemma~\ref{lem:no-upward}.
\end{proof}

\section{Putting the pieces together}
\label{sec:wrap-up}

We assemble the algorithm and bound its cost. The previous sections reduced the problem to one subproblem per red arc: for the head $\theta$ of the arc, find a maximum clique among the blue arcs and the red arcs containing $\theta$. The largest of these cliques, compared with the clique of all blue arcs, gives a maximum clique of the graph.

\paragraph{Preprocessing.}
We fix the cut point $\phi$ and bring the model to its normalized form (Section~\ref{sec:normalized}) in $O(n \log n)$ time, sorting the endpoints. We then build the point sets $R$ and $B$ in $O(n)$ time, the data structure of Lemma~\ref{lem:all-ell} on them in $O(n \log^2 n)$ time, and the sorted lists of the tails and of the heads of the blue arcs. A blue arc contains $\theta$ exactly when its tail exceeds $\theta$ or its head precedes it, so each count $f_\theta$ amounts to two predecessor searches in these lists.

\paragraph{The sweep.}
For the head $\theta$ of every red arc we query the quadrant with $X = \theta$ and $Y = 2n - \theta$, obtaining $\mu_\theta$ by Lemma~\ref{lem:all-ell}, and we add the count $f_\theta$. By Lemma~\ref{lem:visibility}, the sum is the size of a maximum clique among the blue arcs and the red arcs containing $\theta$. We keep the head $\theta^\star$ with the largest sum, and the maximum between this sum and the number of blue arcs is the size of a maximum clique of the graph. Each query costs $O(\log n)$ time, so the sweep takes $O(n \log n)$ time.

\paragraph{Recovering a maximum clique.}
If the maximum clique is the one formed by all the blue arcs, we list them in $O(n)$ time. Otherwise we recover a clique of size $\mu_{\theta^\star} + f_{\theta^\star}$ with the algorithm of Apostolico and Hambrusch~\cite{ApostolicoH87}. Given two points on the circle, that algorithm computes a maximum clique among the arcs that contain at least one of them: the arcs containing both points are taken directly, while among the arcs containing exactly one point a maximum compatible selection is found by a left-to-right scan of their \emph{rollout}, obtained by cutting the circle open at one of the two points so that the arcs become segments on a line. This last step is exactly the maximum compatible bichromatic set problem. Run on the cut point $\phi$ and on $\theta^\star$, it returns a maximum clique among the blue arcs and the red arcs containing $\theta^\star$: the arcs containing both points are the blue arcs containing $\theta^\star$, and the arcs containing exactly one are those of the quadrant at $\theta^\star$. It runs once, in $O(n \log\log n)$ time (Section~3 of~\cite{ApostolicoH87}), and is dominated by the preprocessing.

\thmAlgCircularArcs*
\begin{proof}
Correctness is collected above: the sweep considers, for the head of every red arc, a maximum clique among the blue arcs and the red arcs through it, and one of these cliques, or the clique of all blue arcs, is a maximum clique of the graph. The preprocessing takes $O(n \log^2 n)$ time, dominated by Lemma~\ref{lem:all-ell}, the sweep takes $O(n \log n)$ time, and the recovery takes $O(n \log\log n)$ time.
\end{proof}

The bound depends only on the number of arcs, not on the number of edges $m$, and improves on the $O(m + n \log n)$ algorithm of Bhattacharya and Kaller~\cite{BhattacharyaK97} whenever the graph is dense.

\section{The weighted case}
\label{sec:weighted}

We now consider the weighted version of the problem, where each arc $a$ has a weight $w(a)$, the weight of a clique is the sum of the weights of its arcs, and we seek a clique of maximum weight. We assume the weights are nonnegative without loss of generality. Carrying the weight of each arc to its point, the reduction of the previous sections turns this into the \emph{weighted} maximum compatible bichromatic set problem, identical to the one of Section~\ref{sec:mcbs} except that each point carries a weight and we ask for a compatible subset of maximum total weight in place of one of maximum size. Most steps of the algorithm carry over unchanged, but whether it can be solved as efficiently is open.

The obstacle lies in the shortest path machinery. The score of a staircase becomes the total weight of the blue points on or above it plus the total weight of the red points on or below it, and Lemma~\ref{lem:staircase} holds with the same proof. The grid graph is reweighted accordingly: a rightward edge that had weight $1$ now has the weight $w(q)$ of the crossed point $q$, and an upward edge of the augmented graph has the weight of the point whose horizontal line it crosses. Everything then carries over with sums of weights in place of counts. The weight of a path becomes the total weight of the crossed points it fails to score, so the derivation of Lemma~\ref{lem:all-ell} returns the maximum weight of a compatible subset of the quadrant, with $|Q|$ and $|F|$ replaced by the corresponding total weights. The matrices remain Monge, as the planarity argument does not involve the weights.

What fails is the bound on the core. The argument of Fact~\ref{fact:monge-core} does not work anymore, since the core sum can now be as large as the total weight of the points, and without a core bound the min-plus products of Section~\ref{sec:dc} lose every efficiency guarantee. We therefore do not know whether it is possible to build the matrices efficiently. Cores of size $O(n^{2-\varepsilon})$ would already yield a truly subquadratic algorithm, but proving such a bound appears hard: we show that, already for the across matrix, it would give a faster data structure for weighted range LIS, which is an open problem~\cite{GawrychowskiGK25}.

We now give the reduction. For a sequence $a_1, \ldots, a_n$ of integers with nonnegative weights $w_1, \ldots, w_n$, weighted range LIS asks to preprocess the sequence so that, given an interval of positions, the maximum total weight of a subsequence with non-decreasing values within that interval is returned.

\begin{theorem}
\label{thm:lis-barrier}
Suppose that a data structure, built in $f(n)$ time on a weighted instance of the maximum compatible bichromatic set problem with $n$ points, returns in $g(n)$ time the maximum weight of a compatible subset of the points with $y$-coordinate in a query interval. Then weighted range LIS admits a data structure with preprocessing time $f(2n) + O(n)$ and query time $g(2n) + O(1)$.
\end{theorem}
\begin{proof}
For each index $i$ we place a blue point $B_i$ and a red point $R_i$ at the same location $(a_i,\, n-i)$, both of weight $w_i$: the $x$-coordinate carries the value, and the $y$-coordinate decreases with the index. The points of the indices in an interval $[l, r]$ are then exactly those with $y$-coordinate in $[n-r,\, n-l]$. We claim that the maximum weight of a compatible subset of these points is $\sum_{i=l}^{r} w_i$ plus the answer of the LIS query on $[l, r]$. The reduction follows: we build the $2n$ points and the prefix sums of the weights in $O(n)$ time, and an LIS query is then answered by one interval query, subtracting the baseline $\sum_{i=l}^{r} w_i$, read from the prefix sums.

Consider a compatible subset of the points of the indices in $[l, r]$, and let $I$ be the set of indices with both points in it. For $p < q$ in $I$, the red point $R_p$ lies above the blue point $B_q$, so compatibility forces $a_p \le a_q$, and the values along $I$ are non-decreasing. Since every other index contributes with at most one point, the weight is at most $\sum_{i=l}^{r} w_i + \sum_{i \in I} w_i$.

Conversely, let $I \subseteq [l, r]$ have non-decreasing values, and for an index $i$ let $m_i$ be the largest value associated to an index of $I$ at most $i$, when there is one. Take both points of every index of $I$ and, for every other index of $[l, r]$, its red point when $m_i$ exists and $a_i < m_i$, and its blue point otherwise. Every red point taken has $a_p \le m_p$, and every blue point taken has $a_q \ge m_q$ whenever $m_q$ exists. A red point can dominate only a blue point below it, that is $R_p$ against $B_q$ with $p < q$, since otherwise $R_p$ does not lie above $B_q$. For such a pair $m_p$ exists, hence so does $m_q$, and $a_p \le m_p \le m_q \le a_q$, so the pair is safe. The set is therefore compatible, of weight $\sum_{i=l}^{r} w_i + \sum_{i \in I} w_i$, which proves the claim.
\end{proof}

Our structure answers exactly these interval queries. After the reassignment of Lemma~\ref{lem:distinct}, the interval is taken with endpoints strictly between the point heights, the across matrix of Lemma~\ref{lem:across-queries} is read at the grid rows of the two endpoints, and the false penalties, here the red points above the interval and the blue points below it, are corrected through prefix sums in a similar way to what is done in Lemma~\ref{lem:all-ell}. A core of size $O(n^{2-\varepsilon})$ for the across matrix of every weighted instance would therefore give weighted range LIS data structure with truly subquadratic preprocessing and polylogarithmic query time.

\section{Parametrization by minimum load}
\label{sec:parameterized}

Let $C$ be a clique of $\mathcal{S}$, the set of arcs of the model, and let $a\in C$ be an arc which does not contain any other arc of $C$. Following Apostolico and Hambrusch~\cite{ApostolicoH87}, we say that $a$ is a \emph{base} for $C$ and we call $C$ a clique of base $a$. Every nonempty clique has a base, since containment is a partial order on its finite set of arcs.

Shih and Hsu~\cite{ShihH89} find a maximum weight clique by trying every arc as a base. After sorting the endpoints in $O(n\log n)$ time, they compute for each arc $a$ a heaviest clique having $a$ as a base in $O(\deg(a)\log\log n)$ time, where $\deg(a)$ is the number of arcs intersecting $a$, and they keep the heaviest such clique.
Since $\sum_a \deg(a) = 2m$, where $m$ is the number of intersecting pairs of arcs, computing a heaviest clique for every base takes $O(m\log\log n)$ time, and with the sort the total is $O(n\log n + m\log\log n)$, to our knowledge still the fastest bound for the weighted problem. We show that it suffices to fix a single point on the circle, chosen to be met by as few arcs as possible, and to consider only the arcs intersecting this point as a base.

\begin{lemma}
    \label{lem:base-dichotomy}
    Given a clique $C$ of $\mathcal{S}$, one of the two following conditions holds:
    \begin{enumerate}[label=\normalfont(\roman*)]
        \item  \label{thm:cond1} There is a point of the circle contained in every arc of $C$.
        \item \label{thm:cond2} Every point of the circumference is covered by at least one base of $C$.
    \end{enumerate}
\end{lemma}
\begin{proof}
    Suppose condition~\ref{thm:cond2} does not hold, i.e., some point $\phi$ on the circumference is not covered by any base arc of $C$. We show that in this case condition~\ref{thm:cond1} holds.

    First, we argue that every arc of $C$ contains a base arc of $C$. Indeed, an arc that is not itself a base arc by definition contains another arc of $C$, and the chain of containments ends at a base arc inside the original one.

    The base arcs of $C$ pairwise intersect, as they belong to the clique $C$, and none of them contains $\phi$. Cutting the circle open at $\phi$ turns each base arc into an interval on a line and, since no base arc meets $\phi$, leaves every base arc and every intersection between any two of them untouched. By the Helly property for real-line intervals, the base arcs share a common point. Every arc of $C$ contains a base arc, so this common point lies in every arc of $C$, and hence condition~\ref{thm:cond1} holds.
\end{proof}

The lemma lets us reason about a single point. For a clique whose arcs share a common point we can find the heaviest one directly. For any other clique every point of the circle is covered by one of its bases, so it suffices to enumerate the bases through a single point, and we choose the point covered by the fewest arcs to keep that enumeration small. The running time then depends on that minimum, which we now make precise.

\subsection{The pivot point}

For a point $\theta$ of the circle, let $\Gamma(\theta)$ be the set of arcs containing $\theta$. Any two arcs containing $\theta$ intersect at $\theta$, so $\Gamma(\theta)$ is a clique. We call $d(\theta) = |\Gamma(\theta)|$ the \emph{load} of $\theta$ and $w(\Gamma(\theta))=\sum_{a\in\Gamma(\theta)}w(a)$ its \emph{weighted load}.
Let
\[
    d \,=\, \min_\theta d(\theta)
\]
be the minimum load, attained at a point we fix and call the \emph{pivot} $\psi$. We write $\Gamma(\psi)$ for the set of $d$ arcs through it.
\begin{lemma}
\label{lem:sweep}
The pivot $\psi$, the set $\Gamma(\psi)$, and a point of maximum weighted load can be computed in $O(n)$ time once the endpoints are sorted, which costs $O(n\log n)$.
\end{lemma}
\begin{proof}
The load is constant on each of the $2n$ gaps between consecutive endpoints and changes by one at each endpoint, so a single sweep of the circle reports every load value, and with it the minimum load, the pivot $\psi$, and $\Gamma(\psi)$. The weighted load $w(\Gamma(\theta))$ likewise changes by one arc weight at each endpoint, so the same sweep reports a point of maximum weighted load. The sweep takes $O(n)$ time once the endpoints are sorted, and the sort costs $O(n\log n)$.
\end{proof}

\subsection{Algorithm and analysis}

\thmParamCircularArcs*

\begin{proof}
The algorithm computes two candidate cliques and returns the heavier one. By Lemma~\ref{lem:sweep} it first finds the pivot $\psi$, the set $\Gamma(\psi)$, and a point $\theta^\star$ of maximum weighted load. The first candidate is the clique $\Gamma(\theta^\star)$. For the second, it runs the procedure of Shih and Hsu~\cite{ShihH89} on each of the $d$ arcs $a\in\Gamma(\psi)$, obtaining a heaviest clique having $a$ as a base, and keeps the heaviest over all such $a$. The output is whichever candidate is heavier, together with the clique attaining it.

We first prove the output has maximum weight. Both candidates are cliques, so neither exceeds the weight of a maximum weight clique $C^\star$, and it remains to show that one of them has weight at least $w(C^\star)$. Apply Lemma~\ref{lem:base-dichotomy} to $C^\star$. In case~\ref{thm:cond1} some point $\theta$ lies in every arc of $C^\star$, so $C^\star\subseteq\Gamma(\theta)$. Since the weights are nonnegative and $\theta^\star$ has maximum weighted load, $w(C^\star)\le w(\Gamma(\theta))\le w(\Gamma(\theta^\star))$. In case~\ref{thm:cond2} every point of the circle is covered by a base of $C^\star$, in particular the pivot $\psi$, so some arc of $\Gamma(\psi)$ is a base of $C^\star$ and the enumeration over $\Gamma(\psi)$ returns a clique of weight $w(C^\star)$. Either way a candidate attains $w(C^\star)$.

We now bound the running time. By Lemma~\ref{lem:sweep} the pivot $\psi$, the set $\Gamma(\psi)$, and the point $\theta^\star$ are found in $O(n\log n)$ time. The procedure of Shih and Hsu computes a heaviest clique of base $a$ in $O(\deg(a)\log\log n)$ time, where $\deg(a)$ is the number of arcs intersecting $a$. Running it over the arcs of $\Gamma(\psi)$ takes $O(\sum_{a \in \Gamma(\psi)} \deg(a) \log\log n)$ time, and the sum is at most $dn$, as $\Gamma(\psi)$ holds $d$ arcs of degree less than $n$. The total running time is $O(dn\log\log n + n\log n)$.
\end{proof}

Since $dn \le 2m$, as observed in the introduction, the bound of Theorem~\ref{thm:parameterized} never exceeds the $O(n\log n + m\log\log n)$ of Shih and Hsu, and the advantage grows as $d$ shrinks. For $d = \Theta(n)$ every point of the circle is covered by $\Theta(n)$ arcs, the graph has $\Theta(n^2)$ edges, and the two bounds meet at $O(n^2 \log\log n)$, while for $d = 0$ the model is an interval model, the sweep alone solves the problem, and the bound drops to $O(n \log n)$. With unit weights the bound also holds for the unweighted problem, where it improves on the algorithm of Theorem~\ref{thm:clique-circ-arc} when $d = o(\log^2 n / \log\log n)$.

\section{Maximum clique in interval filament graphs}
\label{sec:filaments}

We turn to interval filament graphs. An \emph{interval filament} $f_u$ is a curve from $(\ell_u, 0)$ to $(r_u, 0)$, attached to an interval $i_u = [\ell_u, r_u]$ of the $x$-axis, drawn in the upper half plane and confined to the vertical strip of $i_u$. That means that two filaments whose intervals are disjoint never intersect. The intersection graph of a family $\mathcal{F} = \{f_1, \ldots, f_n\}$ of interval filaments is an \emph{interval filament graph}. This class contains the interval, circle, chordal, circle-polygon and co-comparability graphs~\cite{Gavril00}.

As in earlier work on these graphs~\cite{Gavril00, BestW21}, the input consists of (1) the explicit intersection graph of the curves, and (2) the interval corresponding to the curve of each vertex in the intersection graph. The abstract graph does not determine a filament representation, and the only geometric data the algorithms ever use are the interval endpoints. We assume, without loss of generality, that the $2n$ endpoints are distinct.

Gavril's algorithm~\cite{Gavril00} finds a maximum weight clique in time $O(n^4)$. It can be improved to $n^{3+o(1)}$ by replacing the black-box invocations of a bipartite matching algorithm with the recent almost-linear time max flow algorithm~\cite{ChenKLPGS25}. We show that, in the unweighted case, a maximum clique can be found in time $O(n^\omega)$, where $\omega < 2.372$ is the exponent of matrix multiplication~\cite{AlmanDWXXZ25}.

\subsection{Gavril's reduction}

We recall the reduction of Gavril~\cite{Gavril00}. Any two filaments of a clique intersect, so their intervals intersect as well, and since intervals have the Helly property all the intervals of a clique share a common point. Every clique is therefore contained in one of the sets
\[
    V_p = \{u : p \in i_u\},
\]
the filaments whose interval contains the point $p$ of the $x$-axis.

Fix a point $p$ and consider two filaments of $V_p$ that do not intersect. Their intervals are not disjoint, as both contain $p$. They do not \emph{overlap} (i.e., intersect with neither containing the other) either, since any two filaments with overlapping intervals must cross. The two intervals are therefore nested. Write $u \prec v$ when the filaments of $u$ and $v$ do not intersect and $i_u \subset i_v$. Gavril proves that $\prec$ is a partial order on $V_p$~\cite{Gavril00}. In particular, when $u \prec v$ and $v \prec w$, the filaments of $u$ and $w$ do not intersect either.

The \emph{comparability graph} of a partial order has the elements as vertices and the comparable pairs as edges. By the previous paragraph, the comparability graph of $\prec$ is exactly the complement of the intersection graph induced by $V_p$: two filaments of $V_p$ intersect precisely when they are incomparable. A maximum clique of $\mathcal{F}$ inside $V_p$ is therefore a maximum \emph{antichain} of $\prec$, a largest set of pairwise incomparable elements, and with weights on the filaments a maximum weight clique is a maximum weight antichain.

The set $V_p$ changes only when $p$ crosses an endpoint, so $O(n)$ choices of $p$, one per gap between consecutive endpoints, capture every clique. Gavril computes a maximum weight antichain of each of the $O(n)$ orders from scratch, by a minimum cost flow computation taking $O(n^3)$ time per order, and the $O(n^4)$ bound follows.

In the unweighted case, computing a maximum antichain is a well-known application of bipartite matching, going back to Fulkerson's proof of Dilworth's theorem~\cite{Fulkerson56}. For a partial order $\prec$ on a finite set $V$, let $B$ be the bipartite graph with a left vertex $u^-$ and a right vertex $u^+$ for every $u \in V$, and an edge $u^- v^+$ whenever $u \prec v$. We write $B_p$ for the graph of the order on $V_p$.

\begin{lemma}
\label{lem:if-matching}
A maximum antichain of $\prec$ has size $|V| - \nu(B)$, where $\nu(B)$ is the size of a maximum matching of $B$.
\end{lemma}
\begin{proof}
In a matching of $B$, every element of $V$ has at most one matched successor and at most one matched predecessor, so the matched pairs decompose $V$ into chains, a matching of size $\nu$ leaving $|V| - \nu$ of them. Conversely, the consecutive pairs of a partition of $V$ into $k$ chains form a matching of size $|V| - k$. Minimum chain partitions therefore correspond to maximum matchings, and by Dilworth's theorem the minimum number of chains partitioning $V$ equals the size of a maximum antichain.
\end{proof}

\subsection{Unweighted cliques via dynamic matching}

We run the sweep as a single dynamic process. Order the $2n$ endpoints from left to right. As $p$ moves along the $x$-axis, a filament $u$ enters $V_p$ when $p$ passes $l_u$ and leaves when $p$ passes $r_u$, so the sweep consists of exactly $2n$ insertions and deletions of elements of the order. The comparabilities of an entering filament are found in $O(n)$ time, comparing its interval with those of the current members of $V_p$ and querying the intersection oracle, for $O(n^2)$ time over the whole sweep. The sequence of updates is the sorted endpoint list, known before the sweep starts, so the maintenance problem we face is offline, and the next lemma solves it.

\begin{lemma}
\label{lem:offline-antichain}
Let a partial order over at most $n$ elements undergo a sequence of $O(n)$ insertions and deletions of elements, each insertion carrying the comparabilities of its element, with the whole sequence known in advance. The size of a maximum antichain after every update can be maintained in $O(n^\omega)$ total time, by a randomized algorithm correct with high probability.
\end{lemma}
\begin{proof}
We maintain the bipartite graph $B$ of the matching formulation and the size of its maximum matching, which by Lemma~\ref{lem:if-matching} determines the size of a maximum antichain. An insertion or deletion of an element amounts to two vertex updates of $B$, one on each side, where a vertex update inserts or deletes one vertex together with its incident edges. The data structure of van den Brand, Forster, Nazari, and Polak~\cite[Theorem 5.1]{BrandFNP24} maintains the maximum matching size under vertex updates with a predicted update sequence, and our sequence is known in advance, so the predictions are exact and each update takes $O(n^{\omega - 1})$ time. Over the $O(n)$ vertex updates this is $O(n^\omega)$ time in total, including the $O(n^\omega)$ initialization. The maintenance is randomized and correct with high probability, as the matching size is read from the rank of a matrix with random entries.
\end{proof}

\thmAlgIntervalFilament*
\begin{proof}
The sweep feeds Lemma~\ref{lem:offline-antichain} with the $2n$ updates, so after every event we hold the size of a maximum antichain of the order on the current $V_p$, that is, the size of a maximum clique inside $V_p$, and we keep the point $p^\star$ attaining the largest one. The $O(n)$ gaps capture every clique, so the recorded maximum is the clique number.

It remains to produce a clique of that size. At the single point $p^\star$ we recompute an actual maximum matching of $B_{p^\star}$ by one maximum flow computation, in $n^{2 + o(1)}$ time~\cite{ChenKLPGS25}, and read off a maximum antichain from it through the constructive form of Dilworth's theorem. This antichain is a maximum clique. The extraction is dominated by $O(n^\omega)$, so the whole algorithm runs in $O(n^\omega)$ time. The only randomized step is the matching size, maintained algebraically and correct with high probability, so the same holds for the algorithm.
\end{proof}

This improves the $O(n^4)$ bound of Gavril~\cite{Gavril00} in the unweighted case.

\paragraph{The weighted case.}
For completeness we record the bound that recent flow algorithms yield for the weighted problem. Gavril solves each of the $O(n)$ comparability instances as a minimum cost flow, computing a maximum weight antichain on a network with $O(n)$ vertices and $O(n^2)$ edges. Plugging in the almost linear time minimum cost flow of~\cite{ChenKLPGS25}, each instance is solved in time $n^{2 + o(1)}$, so a maximum weight clique is found in time $n^{3 + o(1)}$. The gain comes entirely from faster flow, with Gavril's reduction left unchanged.

\section{Conditional lower bounds}
\label{sec:lower-bounds}

We complement our algorithm for interval filament graphs with a conditional lower bound, by reduction from the \emph{orthogonal vectors} problem (OV). This problem asks, given two sets $A, B \subseteq \{0,1\}^d$ of $N$ vectors each, whether there is a pair $a \in A$, $b \in B$ that is orthogonal, i.e., $\sum_{i=1}^d a_i b_i = 0$. Under the Orthogonal Vectors Hypothesis (OVH), which is implied by the Strong Exponential Time Hypothesis (SETH)~\cite{Williams05}, this problem admits no $O(N^{2-\epsilon})$-time algorithm, for any constant $\epsilon > 0$, when $d = c\log N$, for a constant $c$ that depends only on $\epsilon$.

We prove the bound for co-comparability graphs of posets, a subclass of the interval filament graphs. Establishing the bound for a smaller class is a stronger result, and on the way it gives a similar bound for the class of $k$-trapezoid graphs, which sits between the other two. Our lower bounds hold already for posets of dimension $c \log n$, $k$-trapezoid graphs with $k=c \log n$, and interval filament graphs where each filament is a polygonal line with only $c \log n$ segments. As a matter of fact, we show not only that finding a maximum clique is hard in these classes, but that quadratic time is already required to merely distinguish between the case where the entire (implicitly given) graph is a clique and the case where the maximum clique size is at most $n-1$.

\subsection{Co-comparability graphs}

For points $x, y \in \mathbb{R}^k$, we say $x$ is \emph{dominated} by $y$, written $x \le y$, when $x_i \le y_i$ in every coordinate $i \in [k]$.
A poset has \emph{dimension} at most $k$ when it is isomorphic to a subset of points of $\mathbb{R}^k$ under this order. Two points are \emph{incomparable} when neither dominates the other. The \emph{co-comparability graph} of a point set $P$ is the graph $G = (P, E)$ where $E$ is the set of incomparable pairs. So, a clique of $G$ is a set of pairwise incomparable points, an antichain of the order.

\begin{theorem}
Unless OVH fails, there is no $O(n^{2-\epsilon})$-time algorithm that, given a $k$-dimensional poset $P \subseteq \mathbb{R}^k$ of $n$ points, decides if the co-comparability graph of $P$ is a clique, for any constant $\epsilon > 0$ and dimension $k = c\log n$, for a constant $c$ that depends only on $\epsilon$.    
\end{theorem}

\begin{proof}
We prove the theorem by showing that a $T(n,k)$-time algorithm for maximum clique in co-comparability graphs of dimension-$k$ posets would imply an $O(T(2N,2d)+Nd)$-time algorithm for OV.
We encode an OV instance $A,B \subseteq \{0,1\}^d$ as the point set $\{\hat a : a \in A\} \cup \{\hat b : b \in B\}$ in $\mathbb{R}^{2d}$, where
\[
    \hat a := (a_1, \ldots, a_d,1-a_1,\ldots,1-a_d), \qquad \hat b := (1 - b_1, \ldots, 1 - b_d,2 + b_1, \ldots, 2 + b_d).
\]
The point set above has $n = 2N$ points in dimension $k = 2d$, and it can be built in time $O(Nd)$.

For any two vectors from $A$ that differ in the $i$-th coordinate, their corresponding points compare in the two opposite directions in coordinates $i$ and $d+i$, hence they are incomparable.
The same argument applies to any two vectors from $B$.
A point $\hat b$ is never dominated by a point $\hat a$, since their last $d$ coordinates are in $\{2,3\}$ and $\{0,1\}$, respectively. Finally, $\hat a \le \hat b$ asks for $a_i \le 1 - b_i$, for every coordinate $i \in [d]$, that is $a_i + b_i \le 1$, which fails exactly when $a_i = b_i = 1$. So $\hat a \le \hat b$ if and only if no coordinate $i \in [d]$ has $a_i = b_i = 1$, that is if and only if $a$ and $b$ are orthogonal.

Hence, the entire co-comparability graph of this poset is a clique if and only if no pair $a, b$ is orthogonal; otherwise the maximum clique has size at most $n - 1$.
\end{proof}

As a direct consequence we establish the following theorem from the introduction.

\thmLBCocomp*

\subsection{K-trapezoid graphs}

As a corollary of the previous theorems, we obtain a lower bound for maximum clique in $k$-trapezoid graphs, which are a superclass of co-comparability graphs of dimension-$k$ posets. Let $L_1, \ldots, L_{k}$ be parallel lines. A \emph{$k$-trapezoid} picks an interval on each line, and joins consecutive intervals into a polygon across the lines. An intersection graph of a family of $k$-trapezoids is a \emph{$k$-trapezoid graph}.

\corLBktrapezoid*

\begin{proof}
The lower bound follows from the fact that every co-comparability graph of a $k$-dimensional poset is a $k$-trapezoid graph, and that the input poset model can be converted to a trapezoid model in time $O(nk)$. Hence, a subquadratic algorithm for trapezoids would imply a subquadratic algorithm for posets and OVH would fail by \Cref{thm:poset-lower}.

The main idea behind this conversion is very simple: we would like to interpret each point of $\mathbb{R}^k$ as a degenerate zero-width $k$-trapezoid such that, for every $i \in [k]$, its zero-width interval on line $L_i$ is the single value of the point's $i$-th coordinate. Two such trapezoids are disjoint exactly when one lies entirely to the left of the other on every line, that is when one point \emph{strictly} dominates the other. However, in order to handle ties correctly, we need a bit more careful approach.

Given a poset $P \subseteq \mathbb{R}^k$, we first replace each coordinate of each point by its rank among distinct values appearing in that coordinate. This requires sorting $n$ values for each of the $k$ dimensions, which takes time $O(kn \log n)$ and dominates the running time of the conversion. This step does not change pairwise comparisons between elements of $P$, but now we can assume that $P\subseteq [n]^k$.

In the second step, we take each point $p=(p_1, p_2, \dots, p_k) \in P$, compute the sum of its coordinates $s=p_1 + p_2 + \cdots + p_k$ and replace each $p_i$ by $nkp_i + s$. This also does not change pairwise comparisons, but now we have the guarantee that each dominance is a strict dominance, i.e., for distinct $p, q$, whenever $p \leq q$ then $p_i < q_i$ for every $i \in [k]$.

In the final step, we convert each point $p=(p_1, p_2, \dots, p_k) \in P$ to a $k$-trapezoid with interval $[2p_i, 2p_i+1]$ on line $L_i$. Two such trapezoids intersect exactly when the two corresponding points are incomparable.

\end{proof}

\subsection{Interval filament graphs}

The class of interval filament graphs contain the class of $k$-trapezoid graphs~\cite{FelsnerMW97,Gavril00}.
However, to argue that the lower bound for $k$-trapezoids graphs (Corollary~\ref{cor:lb-trapezoid}) implies the lower bound for interval filament graphs (Corollary~\ref{cor:lb-interval-filament}) we need to argue that every $k$-trapezoid graph can be represented as an interval filament graph where each filament is a polygonal line of only $O(k)$ segments.
To this end we design (\Cref{fig:filaments}) a map $\Phi$ that assigns to each $k$-trapezoid $u=([a_1, b_1], [a_2,b_2], \dots, [a_k, b_k])$ the polygonal line $\Phi(u)$ that goes through points
\[
(1-a_1, 0),
(1, a_1),
(1, b_1),
(2, a_2),
(2, b_2),
(3, a_3),
\dots,
(k, b_k),
(k+b_k, 0).
\]
We assume without loss of generality that $a_i \geq 0$ for every $i \in [k]$. Observe that $\Phi(u)$ is an interval filament of $2k + 1$ segments running in the upper half plane from the x-axis at $1 - a_1$ back to the x-axis at $k + b_k$. The first and last segment have slope $+1$ and $-1$, the two pins that bring the curve down to the axis.

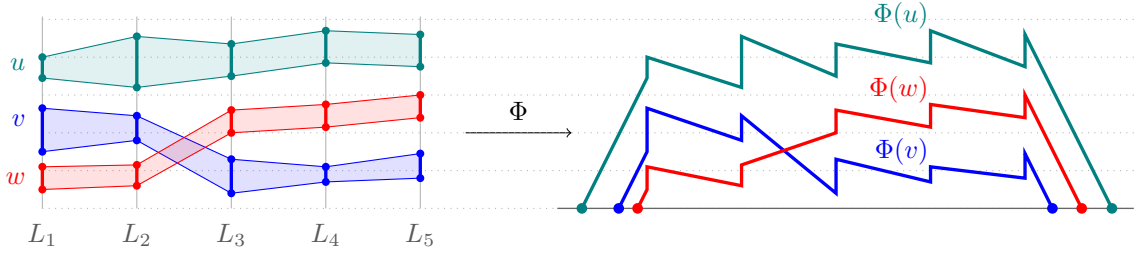
\begin{figure}
    \centering
    \def\RedData{1/0.10/0.22,2/0.12/0.23,3/0.40/0.52,
             4/0.43/0.55,5/0.48/0.60}
\def\BlueData{1/0.30/0.53,2/0.36/0.49,3/0.08/0.26,
              4/0.14/0.22,5/0.16/0.29}
\def\GreenData{1/0.69/0.80,2/0.64/0.91,3/0.70/0.87,
               4/0.77/0.94,5/0.75/0.92}

\newcommand{\DrawTrapezoid}[2]{%
  \foreach \i/\a/\b in #2 {
    \coordinate (#1-a-\i) at (\i,\a);
    \coordinate (#1-b-\i) at (\i,\b);
  }
  \path[draw=#1,fill=#1,fill opacity=0.12]
    (#1-a-1)--(#1-a-2)--(#1-a-3)--(#1-a-4)--(#1-a-5)
    --(#1-b-5)--(#1-b-4)--(#1-b-3)--(#1-b-2)--(#1-b-1)
    --cycle;
  \foreach \i in {1,...,5} {
    \draw[#1,line width=1.25pt] (#1-a-\i)--(#1-b-\i);
    \fill[#1] (#1-a-\i) circle[radius=1.5pt];
    \fill[#1] (#1-b-\i) circle[radius=1.5pt];
  }
}

\newcommand{\DrawFilament}[2]{%
  \foreach \i/\a/\b in #2 {
    \coordinate (#1-p-\i) at (\i,\a);
    \coordinate (#1-q-\i) at (\i,\b);
    \ifnum\i=1\relax
      \coordinate (#1-L) at ({1-\a},0);
    \fi
    \ifnum\i=5\relax
      \coordinate (#1-R) at ({5+\b},0);
    \fi
  }
  \draw[#1,line width=1.25pt]
    (#1-L)--(#1-p-1)--(#1-q-1)
    --(#1-p-2)--(#1-q-2)--(#1-p-3)--(#1-q-3)
    --(#1-p-4)--(#1-q-4)--(#1-p-5)--(#1-q-5)--(#1-R);
  \fill[#1] (#1-L) circle[radius=2pt];
  \fill[#1] (#1-R) circle[radius=2pt];
}

\begin{tikzpicture}[font=\small]
  \draw[->] (6.85,1)--(8.25,1) node[midway,above=1pt] {$\Phi$};

  \foreach \i in {0, 0.5, 1, 1.5, 2, 2.5} {
    \draw[black!30,line width=0.4pt,dotted] (1.25,\i) -- (15.7,\i);
  }

  \begin{scope}[x=1.25cm,y=2.5cm]
  \begin{scope}
    \foreach \i in {1,...,5} {
      \draw[black!30,line width=0.4pt]
        (\i,0)--(\i,1.015);
      \node[anchor=north,text=black!65,inner sep=3pt]
        at (\i,-0.035) {$L_{\i}$};
    }
    \DrawTrapezoid{teal}{\GreenData}
    \DrawTrapezoid{blue}{\BlueData}
    \DrawTrapezoid{red}{\RedData}
    \node[red,anchor=east,inner sep=4pt] at (0.94,0.16) {$w$};
    \node[blue,anchor=east,inner sep=4pt] at (0.94,0.465) {$v$};
    \node[teal,anchor=east,inner sep=4pt] at (0.94,0.765) {$u$};
  \end{scope}

  \begin{scope}[xshift=8cm]
    \draw[black!60,line width=0.5pt] (0.05,0)--(6.15,0);

    \DrawFilament{teal}{\GreenData}
    \DrawFilament{blue}{\BlueData}
    \DrawFilament{red}{\RedData}

    \node[red,anchor=south east,inner sep=0.5pt]
      at (red-q-4) {$\Phi(w)$};
    \node[blue,anchor=south east,inner sep=0.5pt]
      at (blue-q-4) {$\Phi(v)$};
    \node[teal,anchor=south east,inner sep=0.5pt]
      at (teal-q-4) {$\Phi(u)$};
  \end{scope}
  \end{scope}
\end{tikzpicture}
    \caption{The reduction from $k$-trapezoid graphs (here $k=5$) to interval filament graphs, where each filament is a polygonal line with $2k+1$ segments, proving Corollary~\ref{cor:lb-interval-filament}.}
    \label{fig:filaments}
\end{figure}

We claim that, for any two $k$-trapezoids $u=([a_i,b_i])_{i=1}^k$ and $v=([c_i,d_i])_{i=1}^k$, they intersect if and only if the filaments $\Phi(u)$ and $\Phi(v)$ intersect.

We first argue that disjoint trapezoids give disjoint filaments. Two trapezoids are disjoint when their intervals have the same strict order on every line $L_i$ for $i\in [k]$. Assume w.l.o.g.~$b_i < c_i$ for every $i \in [k]$. The vertical segments of $\Phi(u)$ and $\Phi(v)$, i.e., $((i, a_i),(i, b_i))$ and $((i,c_i),(i,d_i))$ are disjoint. Moreover, the segments connecting two consecutive vertical segments are also disjoint, because $((i,b_i),(i+1,a_{i+1}))$ stays strictly below $((i,d_i),(i+1,c_{i+1}))$, because this is true at both endpoints. The leftmost segments of $\Phi(u)$ and $\Phi(v)$ are also disjoint because they lie on distinct parallel lines with slope +1, and the rightmost segments lie on distinct parallel lines of slope -1.

It remains to argue that intersecting trapezoids give intersecting filaments. There are two cases:
\begin{enumerate}
    \item The defining intervals of trapezoids intersect on some line $L_i$. Since the elements contain segments identical to those intervals, they also intersect. 
    \item The intervals are disjoint on every line $L_i$ but their order reverses between two consecutive lines $L_i$ and $L_{i+1}$, i.e., w.l.o.g., $b_i < c_i$ but $d_{i+1} < a_{i+1}$. Then, segments $((i,b_i),(i+1,a_{i+1}))$ and $((i,d_i),(i+1,c_{i+1}))$ intersect.
\end{enumerate}
Having shown that $u$ and $v$ intersect if and only if $\Phi(u)$ and $\Phi(v)$ intersect, we are ready to prove our lower bound.

\corLBIntervalFilament*

\begin{proof}
Consider the following algorithm for the maximum clique problem in $k$-trapezoid graphs. Upon receiving the input trapezoids, it applies $\Phi$ to each of them, which can be done in time $O(nk)$ in total. The intersection graph of the obtained filaments is identical to the intersection graph of the input trapezoids. Hence, it suffices to call the supposed algorithm for maximum clique in interval filament graphs, and returns its answer.

To conclude, a maximum clique algorithm running in time $O(n^{2-\epsilon})$ on interval filament graphs, with filaments given as polygonal lines with $k=c\log n$ segments, implies an $O(n^{2-\epsilon})$-time algorithm for $(\frac{k-1}{2})$-trapezoid graphs, which contradicts OVH by Corollary~\ref{cor:lb-trapezoid}.
\end{proof}

In case that $\omega = 2$, this lower bound matches up to subpolynomial factors the $O(n^\omega)$ upper bound of Theorem~\ref{thm:filament-clique}.

\section*{Acknowledgments}

Generative AI tools (Claude, Codex) were used for preparing initial versions of TikZ figures, and for suggesting improvements in the exposition. The authors assume responsibility for all content.

The work of Adam Polak has been funded by the Italian Ministry of University and Research (MUR) under Ministerial Decree No.\ 18169 of 17 November 2025 -- FIS 3 Call CUP: J53C25002160001. Legally required funding acknowledgment in Italian: \emph{Finanziato con il contributo del Ministero dell’Università e della ricerca ai sensi del D.D.\ n.\ 18169 del 17/11/2025 -- BANDO FIS 3. CUP J53C25002160001.}

\bibliographystyle{alphaurl}
\bibliography{references}

\end{document}